\documentclass[journal]{IEEEtran}

\usepackage{paralist,tikz,graphicx,amsmath, amsfonts,amssymb,amsthm,overpic, latexsym,epsfig,color,rotating,paralist,times,float,subfigure,algorithm,verbatim}
\usetikzlibrary{fit,positioning}

  \newcommand{\beq}{\begin{equation}}
  \newcommand{\eeq}{\end{equation}}
\newcommand{\argmax}{\operatornamewithlimits{arg\,max}}
\newcommand{\argmin}{\operatornamewithlimits{arg\,min}}
\newcommand{\model}{\theta}
\newcommand{\Model}{\Theta}
\newcommand{\mle}{\theta^*}
\newcommand{\modeldim}{M}
\newcommand{\horizon}{N}

\newcommand{\reals}{{\rm I\hspace{-.07cm}R}}
\newcommand{\state}{x}
\newcommand{\obs}{y}
\newcommand{\anoise}{\epsilon}
\newcommand{\belief}{\pi}
\newcommand{\bbelief}{\bar{\belief}}
\newcommand{\dbelief}{\Delta}
\newcommand{\priv}{\eta}
\newcommand{\Belief}{\Pi}
\newcommand{\post}{\alpha}

\newcommand{\obspace}{\mathcal{Y}}
\newcommand{\actionspace}{\mathcal{A}}
\newcommand{\fun}{\phi}
\newcommand{\levels}{g}
\newcommand{\action}{a}
\newcommand{\laction}{u}
\newcommand{\actiondim}{A}
\newcommand{\eaction}{\bar{a}}
\newcommand{\hstate}{\hat{\state}}
\newcommand{\actions}{\{1,\ldots,A\}}
\newcommand{\obss}{\{1,\ldots,\obsdim\}}
\newcommand{\anoisecov}{\sigma^2_\anoise}
\newcommand{\Sig}{K}
\newcommand{\cost}{c}
\newcommand{\tla}{\tilde{\laction}}
\newcommand{\ones}{\mathbf{1}}
\newcommand{\prob}{\mathbb{P}}
\newcommand{\E}{\mathbb{E}}
\newcommand{\lik}{\mathcal{L}}
\newcommand{\nablam}{\nabla_\model}
\newcommand{\unpost}{q}
\newcommand{\unpostm}{\unpost^\model}
\newcommand{\ole}{\stackrel{\text{defn}}{=}}

\newcommand{\kg}{\psi}

\makeatletter
\newcommand{\pushright}[1]{\ifmeasuring@#1\else\omit\hfill$\displaystyle#1$\fi\ignorespaces}
\makeatother

\newcommand{\rate}{\lambda}
\newcommand{\rateo}{\rate_0}

\newcommand{\quantBelief}{\Belief^Q}

\newcommand{\filterd}{\sigma}
\newcommand{\statespace}{\mathcal{X}}
\newcommand{\statedim}{X}
\newcommand{\obsdim}{Y}
\newcommand{\states}{\{1,2,\dots,\statedim\}}
\newcommand{\stateswi}{\{2,\dots,\statedim\}}
\newcommand{\imp}{q}
\newcommand{\oprob}{B}
\newcommand{\aoprob}{R}
\newcommand{\aprob}{G}
\newcommand{\tp}{P}
\newcommand{\pdf}{p}
\newcommand{\statem}{A}
\newcommand{\snoise}{w}
\newcommand{\onoise}{v}
\newcommand{\obsm}{C}
\newcommand{\obsmt}{C^o}
\newcommand{\hobsm}{\obsm}
\newcommand{\snoisecov}{Q}
\newcommand{\onoisecov}{R}
\newcommand{\kalmancov}{\Sigma}
\newcommand{\normal}{N}
\newcommand{\p}{\prime}

\newcommand{\Ga}{\operatorname{Ga}}

\newcommand{\filter}{T}
\newcommand{\nparticles}{N}
\newcommand{\iter}{{(I)}}
\newcommand{\iterI}{I}
\newcommand{\iterplus}{{(I+1)}}

\newtheorem{theorem}            {Theorem}

\newtheorem{definition}         [theorem]{Definition}

\newenvironment{thm}{\begin{theorem}%
  \pushQED{\qed}}%
  {\popQED\end{theorem}}

\newcommand{\weight}{w}

\newcommand{\finaltime}{\horizon}

\newcommand{\innovations}{\iota}

\newcommand{\enemystate}{\hat{\hat{\state}}}
\newcommand{\enemystatem}{\bar{\statem}}
\newcommand{\enemyobsm}{\bar{\obsm}}
\newcommand{\enemykalmancov}{\bar{\kalmancov}}

\newcommand{\enemySig}{\bar{\Sig}}
\newcommand{\enemyonoisecov}{\bar{\onoisecov}}
\newcommand{\enemysnoisecov}{\bar{\snoisecov}}
\newcommand{\enemyinputm}{\bar{F}}
\newcommand{\enemykg}{\bar{\kg}}

\newcommand{\SNR}{\operatorname{SNR}}
\newcommand{\enemySNR}{\overline{\SNR}}
\newcommand{\eobs}{z}
\newcommand{\precision}{\mathcal{P}}

\newcommand{\lR}{\preceq}
\newcommand{\dtime}{n}
\newcommand{\parn}{b}

\newcommand{\mean}{\mu}
\newcommand{\fast}{t}
\newcommand{\fastT}{T}
\newcommand{\copomat}{L}

\newcommand{\var}{\operatorname{Var}}
\newcommand{\tpopt}{\tp^*}
\newcommand{\param}{\theta}
\newcommand{\nablap}{\nabla_\param}

\newcommand{\objective}{J}
\newcommand{\direction}{d}
\newcommand{\tpi}{\tp(1)}
\newcommand{\tpii}{\tp(2)}

\newcommand{\belieftpi}{\belief^{\tpi}}

\newcommand{\belieftpii}{\belief^{\tpii}}
\newcommand{\lr}{\leq_{lr}}
\newcommand{\gr}{\geq_{lr}}
\newcommand{\F}{\mathcal{F}}

  \graphicspath{{./images/}}

\begin{document}

\title{How to Calibrate your Adversary's Capabilities? Inverse Filtering for Counter-Autonomous Systems}

\author{Vikram~Krishnamurthy, {\em Fellow IEEE}  and Muralidhar Rangaswamy,
  {\em Fellow IEEE}\\
  Manuscript dated \today
  \thanks{Vikram Krishnamurthy is  with Cornell Tech and the School of Electrical and Computer Engineering, Cornell University.
    Email: vikramk@cornell.edu.  Muralidhar Rangswamy is with Air Force Research Labs, Dayton, Ohio. Email: muralidhar.rangaswamy@us.af.mil

  This research was
funded by  in part by the Airforce Office of Scientific Research grant FA9550-18-1-0007,
and in part by  the   Army Research Office
research grant 12346080.

A short version of this paper  appeared in the Proceedings of  the International Conference  on Information Fusion, 2019.}}

\maketitle

\begin{abstract}
  We consider an adversarial Bayesian signal processing problem involving ``us'' and an ``adversary''. The adversary observes our  state in noise; updates its posterior distribution of the state and then chooses an action based on this posterior. Given knowledge of ``our'' state and sequence of  adversary's actions observed in noise, we consider three problems: (i) How can the adversary's posterior distribution be estimated?   Estimating the posterior is an inverse filtering problem involving a random measure - we formulate and solve several versions of this problem in a Bayesian setting.
  (ii) How can the adversary's observation likelihood  be estimated?  This tells us how accurate the adversary's sensors are.
 We compute the  maximum likelihood estimator for the adversary's observation likelihood given our measurements of the adversary's actions where the adversary's actions are in response to estimating our state.
(iii) How can the state be chosen by us  to minimize the covariance of the estimate of the adversary's observation likelihood? ``Our'' state can be viewed as a probe signal which causes the adversary to act; so choosing the optimal state sequence is an input design problem.
The above questions are motivated by the design of counter-autonomous systems: given measurements of  the actions of a sophisticated autonomous adversary, how can our counter-autonomous system  estimate the underlying belief of the adversary,  predict future actions and therefore guard against these actions.

\end{abstract}


\IEEEpeerreviewmaketitle

\section{Introduction}

\IEEEPARstart{B}{ayesian} filtering  maps a sequence of noisy observations  to a sequence of posterior distributions of the underlying state.
This paper  considers the inverse problem of reconstructing the posterior given the state and noisy measurements of the posterior; and also estimating the observation likelihood parameter and designing the state signal.
That is we wish to construct an optimal filter to estimate the adversary's optimal filtered estimate, given noisy measurements of the adversary's action.
Such problems arise in  adversarial signal processing  applications involving counter-autonomous systems:  the adversary has a sophisticated autonomous sensing system; given measurements of the adversary's actions we wish to estimate the adversary's sensor's capabilities and predict its future actions (and therefore guard against these actions).

\subsection{Problem Formulation}

The problem formulation involves two players;  we refer to the two players as ``us'' and ``adversary''.  With $k=1,2,\ldots$ denoting discrete time, the model has
the following dynamics:
\begin{equation}
\begin{split}
    \state_k &\sim  \tp_{\state_{k-1},\state} = \pdf(\state | \state_{k-1}), \quad \state_0 \sim \belief_0 \\
    \obs_k  &\sim \oprob_{\state_k,\obs} = \pdf(y | x_k)\\
    \belief_k &= \filter(\belief_{k-1}, \obs_k)\\
    \action_k &\sim \aprob_{\belief_k,\action} = \pdf(\action | \belief_k)
  \end{split} \label{eq:model}
\end{equation}
Let us explain the notation in (\ref{eq:model}):
 $\pdf(\cdot)$ denotes a generic conditional  probability density function (or probability mass function),  $\sim$ denotes distributed according to, and
\begin{compactitem}
    \item $\state_k\in \statespace$ is our Markovian state with  transition kernel $\tp_{\state_{k-1},\state}$ and prior $\belief_0$ where $\statespace$ denotes the state space.
    \item $\obs_k\in \obspace$ is the adversary's noisy observation of our state $\state_k$; with observation likelihoods $\oprob_{xy}$. Here $\obspace$ denote the observation space.
    \item $\belief_k = \pdf(x_k| \obs_{1:k})$ is the adversary's belief (posterior)  of our state $\state_k$ where $\obs_{1:k}$ denotes the sequence  $\obs_1,\ldots,\obs_k$. The operator $T$ in (\ref{eq:model}) is the classical Bayesian filter  \cite{Kri16}
       \beq  \filter(\belief,\obs) = \frac{
    \oprob_{\state ,\obs} \int_\statespace 
    \tp_{\zeta, \state}\, \belief(\zeta) \,d\zeta}
  {\int_\statespace  \oprob_{\state ,\obs} \int_\statespace 
    \tp_{\zeta, \state}\, \belief(\zeta)\, d\zeta d\state}
  \label{eq:belief}
\eeq
Let $\Belief$ denote the  space of all such beliefs. When the state space
$\statespace$ is Euclidean space, then $\Belief$ is a function space comprising the space of density functions; if $\statespace$ is finite, then $\Belief$ is  the unit $
\statedim-1$ dimensional simplex of $\statedim$-dimensional probability mass functions.
    \item $\action_k \in \actionspace$ denotes our measurement of the adversary's action based on its current belief\footnote{More explicitly, the adversary chooses an action $\laction_k$ as a deterministic function of $\belief_k$ and we observe $\laction_k$ in noise as $\action_k$. We encode this as $\aprob_{\belief_k,\action_k}$.} $\belief_k$ where $\actionspace$ denotes the action space. Thus $\aprob$ is the conditional probability (or density if $\actionspace$ is continuum) of observing an action $\action$ given the adversary's belief $\belief$.
    \end{compactitem}

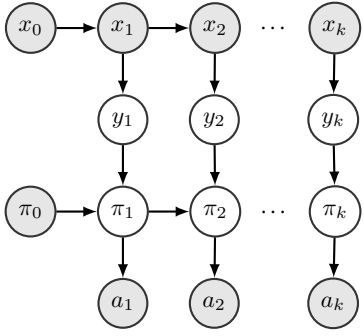
\begin{figure}[h]
\centering
\begin{tikzpicture}[thick,scale=0.9, every node/.style={transform shape}]
\tikzstyle{main}=[circle, minimum size = 5mm, thick, draw =black!80, node distance = 6mm]
\tikzstyle{connect}=[-latex, thick]
  \node[main, fill = black!10] (x0) {$x_0$};
  \node[main, fill = black!10] (x1) [right=of x0] {$x_1$};
  \node[main, fill = black!10] (x2) [right=of x1] {$x_2$};
  \node[main, fill = black!10] (xk) [right=10mm of x2] {$x_k$};
  \path (x0) edge [connect] (x1)
        (x1) edge [connect] (x2)
        (x2) -- node[auto=false]{\ldots} (xk);
  
  \node[main] (y1) [below=of x1] {$y_1$ };
  \node[main] (y2) [right=of y1] {$y_2$};
  \node[main] (yk) [right=10mm of y2] {$y_k$};
  \path (x1) edge [connect] (y1)
        (x2) edge [connect] (y2)
        (xk) edge [connect] (yk);
  
  \node[main] (pi1) [below=of y1] {$\pi_1$ };
  \node[main, fill = black!10] (pi0) [left=of pi1]{$\pi_0$};
  \node[main] (pi2) [right=of pi1] {$\pi_2$};
  \node[main] (pik) [right=10mm of pi2] {$\pi_k$};
  \path (pi0) edge [connect] (pi1)
        (pi1) edge [connect] (pi2)
        (pi2) -- node[auto=false]{\ldots} (pik)
    
        (y1) edge [connect] (pi1)
        (y2) edge [connect] (pi2)
        (yk) edge [connect] (pik);

  \node[main, fill = black!10] (a1) [below=of pi1]{$a_1$};
  \node[main, fill = black!10] (a2) [right=of a1] {$a_2$};
  \node[main, fill = black!10] (ak) [right=10mm of a2] {$a_k$};
  \path (pi1) edge [connect] (a1)
        (pi2) edge [connect] (a2)
        (pik) edge [connect] (ak);
\end{tikzpicture}
\caption{Graphical representation of inverse filtering problem. The shaded nodes are known to us, while the white nodes are unknown. The model is described in (\ref{eq:model}).}
\label{fig:graph}
\end{figure}
    
Figure \ref{fig:graph} is a graphical representation of the model (\ref{eq:model}). The shaded nodes are known to us, while the white nodes are computed by the adversary and unknown to us. This is in contrast to classical Bayesian filtering where $\obs_k$ is known to us and $\state_k$ is to be estimated, i.e, $\belief_k$ is to be computed.

\subsection{Objectives}
Given (\ref{eq:model}), this paper  considers  the following  problems:
\begin{compactenum}
    \item \textbf{Inverse Filtering}: How to estimate the adversary's belief given measurements of its actions (which are based on  its filtered estimate of our state)? Assuming  probability distributions  $\tp,\oprob,\aprob$ are  known, estimate  the adversary's belief $\belief_k$ at each time $k$, by computing   posterior $\pdf(\belief_k\mid \belief_0,\state_{0:k},\action_{1:k})$. From a practical point of view, estimating the adversary's belief allows us to  predict (in a Bayesian sense)  future actions of the adversary. In the design of counter-autonomous systems, this facilitates taking effective measures against such actions.

    \item \textbf{Covariance Bounds in Localization}:  Localization refers to estimating the underlying state $\state$ when  the transition kernel is identity; so  state $\state$ is a random variable.
      We consider the case where both  us and the adversary do not know  $\state$. For the Gaussian case, we consider  a sequential game setting involving us and the adversary estimating the underlying state. We show that the asymptotic covariance is twice as large as that of classical localization.

          \item \textbf{Parameter Estimation}: How to remotely calibrate the adversary's sensor? Assume  $\tp,\aprob$ are known and  $\oprob$ is parametrized  by vector $\theta$. Compute the maximum likelihood estimate (MLE)  of $\theta$ given  $\state_{0:\horizon},\action_{1:\horizon}$ for fixed data length $\horizon$.
            From a practical point of view, this determines the accuracy of the adversary's sensor.
    \item \textbf{Optimal Probing}. The transition kernel
      $\tp$ determines our random signal  $\state$ which  ``probes'' the adversary's sensor. What choice of $\tp$ yields the smallest variance in our estimate of the adversary's likelihood $\oprob$?
\end{compactenum}

\subsection{Examples} \label{sec:examples}
 1. {\em Calibrating a Radar}.
Here ``us'' refers to a drone/UAV  or electromagnetic signal that probes an adversary's sophisticated multi-function  radar system.
The adversary's radar records measurements $\obs_k$ of the drones  kinematic state $\state_k$ and its Bayesian tracking functionality  computes $\belief_k$. The radar resource manager then chooses an action $\laction_k$
(e.g.\ waveform, beam orientation/aperture) and our inference/measurement of this action is
$\action_k$. The above  problems then boil down to: estimating the adversary's tracked estimate so as to predict its future actions (inverse filtering); determining the accuracy of the adversary's radar system (parameter estimation), optimizing the drone's trajectory  to estimate $\oprob$ accurately (optimal probing).
\\
{\em 2. Electronic Warfare}. Suppose $\state_k$ is our electromagnetic signal applied to jam the adversary's radar (e.g.\ a cross polarized signal where the phase evolves  according to a Markov process). The radar responds with various actions. Our aim is to determine the effectiveness of the jamming signal; $\oprob$ models the likelihood of the radar recognizing the jamming signal, $\belief_k$ denotes the radar's belief of whether it is being jammed; and $\action_k$ denotes the action we observe from the radar.

{\em 3. Interactive Learning}. Although we do not pursue it here, the above model has strong parallels with those used in personalized education and interactive learning:  $\state_k$ denotes the quality/difficulty of material being taught; $\oprob$ models how the learner absorbs the material; $\belief_k$ models the learner's knowledge of the material and $\action_k$ is the response to an exam administered by the instructor. The instructor aims to  estimate the learner's knowledge and  the learner's  absorption probabilities to optimize how to present the material.

{\em Simplifications}:
This  paper makes several  simplifying  assumptions. First,
we assume the adversary knows our transition kernel $\tp$ to compute  the Bayesian belief using  (\ref{eq:belief}).
In reality the adversary  needs to estimate $\tp$  based on our trajectory; and we  need to estimate the adversary's estimate of $\tp$. 
Second, we assume that the adversary's sensors are not reactive - that is, it does not know that we are probing it. If the adversary is a reactive  low probability of intercept (LPI) radar, it  will attempt to confuse our estimator resulting in a dynamical  game. Both these simplifying  assumptions point to a  more general game-theoretic  setting; which is the subject of future work.

\subsection{Related Works}
Counter unmanned autonomous systems are discussed in \cite{Kup17}.
Specifically, \cite[Figure 1]{Kup17} places such systems at a level of abstraction above the physical sensors/actuators/weapons and datalink layers; and below the human controller layer.

This  paper extends our recent works \cite{MRKW17,MRKW18,MIC19} where the mapping from
belief $\belief$ to action $\action$ was assumed  deterministic. Specifically, \cite{MRKW17} gives a deterministic regression based approach to  estimate the adversary's model parameters in a Hidden Markov model. In comparison, the current paper assumes a probabilistic map between $\belief $ and $\action$ and develops Bayesian filtering algorithms for estimating  the posterior along with MLE algorithms for $\model$. There strong parallels between inverse filtering and Bayesian social learning  \cite{Cha04} as described in  Sec.\ref{sec:social}; the key difference is that in social learning the aim is to estimate the underlying state given noisy posteriors, whereas our aim is to estimate the posterior given noisy posteriors and the underlying state.

\section{Optimal Inverse Filter for Estimating Belief}
Given the model  (\ref{eq:model}) we now derive a filtering recursion for the posterior of the adversary's belief given knowledge  of our state sequence and recorded actions. Accordingly, define
$$\post_{k}(\belief_k) = \pdf(\belief_k | \action_{1:k},\state_{0:k}).$$
Note that the posterior $\post_k(\cdot)$ is a {\em random measure} since it is a posterior distribution of the adversary's posterior  distribution (belief) $\belief_k$.
\begin{thm} \label{thm:post}
The posterior  $\post_k$ satisfies the following filtering recursion initialized by prior random measure $\alpha_0 = \pdf(\belief_0)$:
\beq
  \post_{k+1}(\belief) = \frac{\aprob_{\belief,\action_{k+1}}
    \,  \int_\Belief \oprob_{\state_{k+1}, \obs_{\belief_k,\belief}}\, \post_k(\belief_k) d\belief_k}
  {\int_\Belief \aprob_{\belief,\action_{k+1}}
    \,  \int_\Belief \oprob_{\state_{k+1}, \obs_{\belief_k,\belief}}\, \post_k(\belief_k) d\belief_k \, d\belief}
  \label{eq:post}
  \eeq
 Here $\obs_{\belief_k,\belief}$ is the observation such that $ \belief = \filter(\belief_k,\obs)$ where $\filter$ is the adversary's  filter (\ref{eq:belief}). The conditional mean estimate of the belief is $\E\{\belief_{k+1}|\action_{1:k},\state_{0:k}\} = \int_\Belief \belief \post_{k+1}(\belief) d\belief $.
\end{thm}
\begin{proof}
Start with the un-normalized density:
\begin{multline*}
  \pdf(\belief_{k+1},\obs_{k+1},\action_{1:k+1}, \state_{0:k+1})
  = \pdf(\action_{k+1}|\belief_{k+1})\, \pdf(\obs_{k+1}|\state_{k+1})\, \\
  \times \pdf(\state_{k+1}|\state_k)\, \int_\Belief \pdf(\belief_{k+1}| \obs_{k+1},\belief_k)\,  \, \pdf(\belief_k,\action_{1:k},\state_{0:k})\, d\belief_k 
\end{multline*}
Note that $\pdf(\belief_{k+1}|\obs_{k+1},\belief_k) = I(\belief_{k+1} - \filter(\belief_k, \obs_{k+1}))$ where $I(\cdot)$ denotes  the indicator function.
Finally,  marginalizing the above  by integrating over $\obs_{k+1}$ and then normalizing yields (\ref{eq:post}). The $\pdf(\state_{k+1}|\state_k)$ term cancels out after normalization.

\end{proof}

  We call (\ref{eq:post})  the {\em optimal inverse filter} since it yields the Bayesian posterior of the adversary's  belief given our state and noisy measurements of the  adversary's actions.
At first sight, it appears that $\post_{k+1}$ does not depend on the transition kernel $\tp$. Actually $\post_{k+1}$   does depend on $\tp$ since  from (\ref{eq:belief}),  $\obs_{\belief_k,\belief_{k+1}}$ depends on $\tp$.

The rest of this section considers  inverse filtering algorithms  for evaluating $\post_k(\belief)$ for the following 5 cases:
\begin{compactenum}
\item  Inverse Hidden Markov filter (finite dimensional but exponential computational cost in time)
\item Inverse Kalman filter (finite dimensional sequence of Kalman filters)
\item Sequential MCMC for inverse filtering; the structure faciliates  using the so called ``optimal importance'' function.
\item Inverse filtering in multi-agent social learning
\item Inverse  Bayesian localization involving conjugate priors (Gaussian likelihood and 
 Gamma prior)
\end{compactenum}

\subsection{Inverse Hidden Markov Model (HMM) Filter}
\label{sec:inversehmm}
 We first illustrate the inverse filter (\ref{eq:post}) for the case where the adversary deploys a HMM filter.
  Suppose $\statespace = \states$,
  $\obspace = \obss$,  $\actionspace = \actions$, implying that our state $\{\state_k, k \geq 0\}$ is a finite state Markov chain with transition matrix $\tp$ and the adversary has  finite valued HMM observations $\{\obs_k,k\geq 0\}$ with  observation probability $\oprob$.
Then $\filter(\belief,\obs)$ in (\ref{eq:belief}) is the classical HMM filter and the belief space $\Belief$ is the unit $\statedim-1$ dimensional simplex, i.e., the space of $\statedim$-dimensional probability vectors.
Suppose (\ref{eq:post}) is initialized  with $\alpha_0(\belief) = \delta(\belief- \belief_0)$, i.e., the prior is a Dirac-delta function placed at $\belief_0 \in \Belief$.

For $k=1,2,\ldots$ construct the  finite sets $\Belief_k$ of belief states via the following recursion:
$$\Belief_k = \big\{ \filter(\belief,\obs), \;\obs \in \obspace, \belief \in \Belief_{k-1} \big\}, \quad  \Belief_0 = \{\belief_0\}.$$
 Note  $\Belief_k$ has
$\obsdim^k$ elements.

Using (\ref{eq:post}), the inverse HMM filter for estimating the belief reads: for $\belief \in \Belief_{k+1}$, the posterior and conditional mean are
\begin{align}
 & \post_{k+1}(\belief) = \frac{\aprob_{\belief,\action_{k+1}}
    \,  \sum_{\bbelief \in \Belief_k} \oprob_{\state_{k+1}, \obs_{\bbelief,\belief}}\, \post_k(\bbelief) }
  {\sum_{\belief \in \Belief_{k+1}} \aprob_{\belief,\action_{k+1}}
    \,  \sum_{\bbelief \in \Belief_k} \oprob_{\state_{k+1}, \obs_{\bbelief,\belief}}\, \post_k(\bbelief) }
\nonumber  \\
 & \E\{\belief_{k+1}|\action_{1:k+1},\state_{0:k+1}\} = \sum_{\belief \in \Belief_{k+1}} \belief \post_{k+1}(\belief)
 \label{eq:hmmpost}
  \end{align}
  The inverse HMM filter (\ref{eq:hmmpost})  is a finite dimensional recursion, but the cardinality  of 
 $\Belief_k$ grows exponentially with $k$.


\subsection{Inverse Kalman Filter} \label{sec:inversekalman}
We consider  a second special case of (\ref{eq:post}) where the  inverse filtering problem  admits a finite dimensional characterization in terms of  the Kalman filter.
Consider a  linear Gaussian state space model
\beq \label{eq:lineargaussian}
\begin{split}
\state_{k+1} &= \statem\, \state_k  + \snoise_k, \quad \state_0 \sim \belief_0 \\
\obs_k &= \obsm\, \state_k + \onoise_k 
\end{split}
\eeq
where  $\state_k \in \statespace = \reals^\statedim$ is ``our'' state with
initial density $\belief_0 \sim \normal(\hat{\state}_0,\kalmancov_0)$,
 $\obs_k \in \obspace = \reals^\obsdim$ denotes the adversary's observations,
 $\snoise_k\sim \normal(0,\snoisecov_k)$,
 $\onoise_k \sim \normal(0,
\onoisecov_k)$
and 
  $\{\snoise_k\}$,  
  $\{\onoise_k\}$ are mutually independent  i.i.d.\ processes.

 Based on observations $\obs_{1:k}$, the adversary computes the  belief  $\belief_k = \normal(\hstate_k,\kalmancov_k)$ where $\hstate_k$ is the conditional mean
  state   estimate and $\kalmancov_k$ is the covariance; these are computed via the classical Kalman filter
  equations:\footnote{For localization problems, we will use the information filter form:
    \beq  \kalmancov^{-1}_{k+1} = \kalmancov_{k+1|k}^{-1} + \obsm^\p \onoisecov^{-1} \obsm, \quad
\kg_{k+1} = \kalmancov_{k+1} \obsm^\p \onoisecov^{-1} \label{eq:info}
\eeq Similarly, the inverse Kalman filter in information form reads
\beq \enemykalmancov^{-1}_{k+1} = \enemykalmancov^{-1}_{k+1|k} +
\enemyobsm_{k+1}^\p \enemyonoisecov^{-1} \enemyobsm_{k+1},\;
\enemykg_{k+1} = \enemykalmancov_{k+1} \enemyobsm_{k+1}^\p \enemyonoisecov^{-1}. \label{eq:enemyinfo}
\eeq}
\beq
  \begin{split}
\kalmancov_{k+1|k} &=  \statem  \kalmancov_{k} \statem^\p  +  \snoisecov  \\
\Sig_{k+1} &= \obsm \kalmancov_{k+1|k} \obsm^\p + \onoisecov 
\\
{\hstate}_{k+1} &=  \statem\,  {\hstate}_k  + 
\kalmancov_{k+1|k} \obsm^{\p}  \Sig_{k+1}^{-1} 
(\obs_{k+1} - \obsm \, \statem\,  {\hstate}_k )
\\
\kalmancov_{k+1} &=
\kalmancov_{k+1|k} -  
\kalmancov_{k+1|k} \obsm^{\p}  \Sig_{k+1}^{-1} 
\obsm \kalmancov_{k+1|k} 
\end{split}
\label{eq:kalman}
\eeq
  The 
  adversary then chooses its  action as  $\eaction_k = \fun(\kalmancov_k)\,\hstate_k$ for some pre-specified function\footnote{This choice is motivated by linear quadratic Gaussian control where the action (control) is chosen as a linear function of the estimated state $\hstate_k$ weighed by the covariance matrix}
  $\fun$. We  measure the adversary's  action as
\beq \action_k = \fun(\kalmancov_k)\,\hstate_k + \anoise_k, \quad
\anoise_k \sim \text{ iid } \normal(0,\anoisecov) \label{eq:linearaction} \eeq

The Kalman covariance $\kalmancov_k$ is deterministic and fully determined by the model parameters. So  to estimate the belief
$\belief_k=\normal(\hstate_k,\kalmancov_k)$ we only need  to estimate $\hstate_k$ at each time $k$ given $a_{1:k},\state_{0:k}$.
Substituting (\ref{eq:lineargaussian})  for $\obs_{k+1}$ in (\ref{eq:kalman}), we see that
(\ref{eq:kalman}), (\ref{eq:linearaction}) constitute a linear Gaussian system with un-observed state  $\hstate_k$, observations $\action_k$,
and known exogenous  input $\state_k$:
\beq \label{eq:inversekf}
\begin{split}
  {\hstate}_{k+1} &=   (I - \kg_{k+1} \obsm) \, \statem \hstate_{k} + \kg_{k+1} \onoise_{k+1} + \kg_{k+1} \obsm \state_{k+1} \\
   \action_k &= \fun(\kalmancov_k)\,\hstate_k + \anoise_k, \quad
   \anoise_k \sim \text{ iid } \normal(0,\anoisecov) \\
  & \text{ where }  \kg_{k+1} = \kalmancov_{k+1|k} \obsm^{\p}  \Sig_{k+1}^{-1} 
\end{split}
\eeq
$
\kg_{k+1}$   is  called the Kalman gain.

To summarize,  our filtered estimate of the adversary's filtered estimate    $\hstate_k$ given measurements $a_{1:k},\state_{0:k}$ is achieved by running ``our''  Kalman filter on the linear Gaussian state space model  (\ref{eq:inversekf}), where
$\hstate_k, \kg_k, \kalmancov_k$ in (\ref{eq:inversekf}) are generated by the adversary's  
Kalman filter. Therefore, our Kalman filter uses the parameters
\beq
\begin{split}
\enemystatem_{k}  &=  (I - \kg_{k+1} \obsm)\statem, \;
\enemyinputm_k = \kg_{k+1} \obsm ,\;
  \enemyobsm_k =  \fun(\kalmancov_k) , \\
\enemysnoisecov_k  &= \kg_{k+1}\, \kg_{k+1}^\p, \;
\enemyonoisecov = \anoisecov
\end{split} \label{eq:inversekfparam}
  \eeq
 The equations of our inverse Kalman filter are:
\beq
 \begin{split}
  \enemykalmancov_{k+1|k} &=  \enemystatem_k \enemykalmancov_{k} \enemystatem_k^\p  +  \enemysnoisecov_k  \\
  \enemySig_{k+1} &= \enemyobsm_{k+1} \enemykalmancov_{k+1|k} \enemyobsm_{k+1}^\p + \enemyonoisecov  \\
   \enemystate_{k+1} &= \enemystatem_k\, \enemystate_k + 
\enemykalmancov_{k+1|k} \enemyobsm_{k+1}^{\p}  \enemySig_{k+1}^{-1} 
                        \\ & \pushright{\times \big[\action_{k+1} -   \enemyobsm_{k+1} \left(\enemystatem_{k} \hstate_k+ \enemyinputm_k \state_{k+1} \right) \big]} 
  \\
   \enemykalmancov_{k+1} &=
\enemykalmancov_{k+1|k} -  
\enemykalmancov_{k+1|k} \enemyobsm_{k+1}^{\p}  \enemySig_{k+1}^{-1} 
\enemyobsm_{k+1} \enemykalmancov_{k+1|k} 
\end{split} \label{eq:inversekfequations}
\eeq
Note $\enemystate_k$ and $\enemykalmancov_k$ denote our conditional mean estimate and covariance of the adversary's conditional mean $\hstate_k$.

  \subsection{Particle Filter for Inverse Filtering}
  In general the optimal inverse filter (\ref{eq:post}) does not have a finite dimensional statistic. Also,  the inverse HMM filter (\ref{eq:hmmpost}) is intractable for large data lengths. One needs to resort to an approximation such as sequential MCMC. Here we  describe a  particle filter that uses the optimal importance function (i.e., minimizes  the variance of the importance weights).
  Due to the different dependency structure compared to classical state space models (see Figure \ref{fig:graph}),
  the  updates for the  importance weights given below are  different
  compared to classical particle filters.
  
  The particle filter constructs the $\nparticles$-point Dirac delta function approximation to the posterior as:\footnote{We assume here that the adversary computes the belief $\belief$ exactly. Of course, in reality, the adversary itself would use a sub-optimal approximation to the optimal filter.}
  $$\pdf(\belief_{0:k},\obs_{1:k}| \state_{0:k},\action_{1:k}) \approx
  \sum_{i=1}^\nparticles \frac{\weight_k^{(i)}(\belief^{(i)}_{0:k},\obs^{(i)}_{1:k})}
  {\sum_{j=1}^\nparticles \weight_k^{(i)}(\belief^{(j)}_{0:k},\obs^{(i)}_{1:k})}
  \delta(\belief^{(i)}_{0:k},\obs^{(i)}_{1:k}) $$

  In complete analogy to the classical particle filter, we construct the sequential importance sampling update as follows:
\begin{multline} \label{eq:numerator}
  \pdf(\belief_{0:k},\obs_{1:k}| \state_{0:k},\action_{1:k}
  )
  \propto
  \pdf(\action_k| \belief_k) \, \pdf(\belief_k| \belief_{k-1},\obs_k)\,
  \pdf(\obs_k|\state_k)\\ \times \pdf(\state_k|\state_{k-1})\,
  \pdf(\belief_{0:k-1},\obs_{1:k-1}| \state_{0:k-1},\action_{1:k-1})
  \end{multline}
  The importance density $\imp(\cdot)$ is chosen such that
  \begin{multline} \label{eq:denominator}
    \imp(\belief_{0:k},\obs_{1:k} | \state_{0:k},\action_{1:k})   =
     \imp(\belief_{0:k-1},\obs_{1:k-1} | \state_{0:k-1},\action_{1:k-1}) 
    \\
  \times   \imp (\belief_k,\obs_k| \belief_{0:k-1}, \obs_{0:k-1},\state_{0:k},\action_{1:k}) 
  \end{multline}
  Then defining the importance weights
  $$\weight_k^{(i)}(\belief_{0:k}^{(i)},\obs_{1:k}^{(i)}) =
  \frac{\pdf(\belief_{0:k}^{(i)},\obs_{1:k}^{(i)}| \state_{0:k},\action_{1:k})} { \imp(\belief_{0:k}^{(i)},\obs_{1:k}^{(i)}| \state_{0:k},\action_{1:k})}$$ (\ref{eq:numerator}), (\ref{eq:denominator}) yield the following time recursion:
  \begin{multline} \label{eq:weight}
    \weight_k^{(i)}(\belief_{0:k}^{(i)},\obs_{1:k}^{(i)}) \propto
    \weight_{k-1}^{(i)}(\belief_{0:k-1}^{(i)},\obs_{1:k-1}^{(i)})  \\
\times  \frac{\aprob_{\belief_k^{(i)},\action_k}\, \delta(\belief_k^{(i)} - \filter(\belief_{k-1}^{(i)},\obs_k^{(i)}))\, \oprob_{\state_k,\obs_k^{(i)}}\, \tp_{\state_{k-1},\state_k}} {
    \imp(\belief_k^{(i)},\obs_k^{(i)}|
    \belief_{0:k-1}^{(i)},\obs_{0:k-1}^{(i)},\state_{0:k},\action_{1:k})}
\end{multline}
In particle filtering, it  recommended to use the so called ``optimal'' importance density \cite{RAG04,CMR05}, namely,  $\imp^* = \pdf(\belief_k,\obs_k| \belief_{k-1},\obs_{k-1}, \state_{k},\action_{k} )$.
Note that in our case $\imp^*=  \pdf(\belief_k | \belief_{k-1},\obs_k)\, \pdf(\obs_k| \state_k,\action_k)
= \delta\big(\belief_k - \filter(\belief_{k-1},\obs_k)\big)\, \pdf(\obs_k|\state_k)$ is straightforward to sample from. Also using $\imp^*$,  the importance weight update
(\ref{eq:weight}) becomes
 \begin{multline} \label{eq:optweight}
\weight_k^{(i)}(\belief_{0:k}^{(i)},\obs_{1:k}^{(i)}) \propto
    \weight_{k-1}^{(i)}(\belief_{0:k-1}^{(i)},\obs_{1:k-1}^{(i)})  \,
    \aprob_{\belief_k,\action_k} \tp_{\state_{k-1},\state_k}
  \end{multline}
  In summary, the particle filtering algorithm becomes:\\
{\em Sequential Importance Sampling step}: At each time $k$
\begin{compactitem}
\item For $i=1,\ldots,\nparticles$ sample from importance density $\imp$ or~$\imp^*$ $$(\tilde{\belief}_k^{(i)},\tilde{\obs}_k^{(i)}) \sim \imp(\tilde{\belief}_k,\obs_k| \belief^{(i)}_{0:k-1},\obs^{(i)}_{1:k-1},\state_{0:k},\action_{1:k})$$
  Set $(\tilde{\belief}_{0:k}^{(i)},\obs_{1:k}^{(i)}) =
  (\belief^{(i)}_{0:k-1},\tilde{\belief}_k^{(i)},\obs^{(i)}_{1:k-1}, \tilde{\obs}_k^{(i)}) $.

 \item  Update importance weights $\weight_k^{(i)}$ using  (\ref{eq:weight}) or (\ref{eq:optweight}).

\end{compactitem}

{\em Resampling/Selection step}: Multiply/Discard particles  with 
high/low normalised importance weights
to obtain $\nparticles$ new particles with unit weight.

Finally, the particle filter estimate of the posterior $\belief_k$ at each time~$k$ is
$\E\{\belief_k| \action_{1:k},\state_{0:k}\} \approx  \sum_{i=1}^\nparticles  \weight_k^{(i)} \belief_k^{(i)}/ \sum_{i=1}^\nparticles  \weight_k^{(i)} $.

To summarize, the particle filter provides a  tractable sub-optimal algorithm  for inverse filtering and it is straightforward to sample from the optimal importance density.

\subsection{Inverse Filtering for Multiagent Social Learning Models} \label{sec:social}

Thus far we have considered inverse filtering based on the graphical  model in Figure \ref{fig:graph}.
Below, motivated  by Bayesian social learning models   in behavioral economics and sociology (which is also now popular in signal processing), we consider inverse filtering for the graphical model depicted in Figure \ref{fig:graphsocial}.
The adversary now is a 
 multi-agent system that aims to estimate the state of ``our" underlying Markov state $\state_k \in
 \states$.
The key difference compared to Figure \ref{fig:graph}  is  that in Figure~\ref{fig:graphsocial}   the previous action taken by the adversary  directly affects the dynamics of the  belief update.  Such social learning is motivated by multi-agent decision making in social groups/networks and can result in spectacular behavior such as  cascades and herding; see \cite{KP14} for an extensive signal processing centric discussion.

\begin{figure}[h]
\centering
\begin{tikzpicture}[thick,scale=0.9, every node/.style={transform shape}]
\tikzstyle{main}=[circle, minimum size = 5mm, thick, draw =black!80, node distance = 6mm]
\tikzstyle{connect}=[-latex, thick]
  \node[main] (x0) {$x_0$};
  \node[main] (x1) [right=of x0] {$x_1$};
  \node[main] (x2) [right=of x1] {$x_2$};
  \node[main] (xk) [right= of x2] {$x_3$};
  \path (x0) edge [connect] (x1)
        (x1) edge [connect] (x2)
        (x2) edge [connect] (xk);
  
  \node[main] (y1) [below=of x1] {$y_1$ };
  \node[main] (y2) [right=of y1] {$y_2$};
  \node[main] (yk) [right=of  y2] {$y_3$};
  \path (x1) edge [connect] (y1)
        (x2) edge [connect] (y2)
        (xk) edge [connect] (yk);
  
  \node[main] (pi1) [below=of y1] {$\pi_1$ };
  \node[main, fill = black!10] (pi0) [left=of pi1]{$\pi_0$};
  \node[main] (pi2) [right=of pi1] {$\pi_2$};
  \node[main] (pik) [right= of pi2] {$\pi_3$};
  \path (pi0) edge [connect] (pi1)
    
        (y1) edge [connect] (pi1)
        (y2) edge [connect] (pi2)
        (yk) edge [connect] (pik);

  \node[main] (a1) [below=of pi1]{$\laction_1$};
  \node[main] (a2) [right=of a1] {$\laction_2$};
  \node[main] (ak) [right= of  a2] {$\laction_3$};
  \node[main, fill = black!10] (u1) [below=of a1]{$\action_1$};
  \node[main, fill = black!10] (u2) [below=of a2]{$\action_2$};
  \node[main, fill = black!10] (u3) [below=of ak]{$\action_3$};
  \path (pi1) edge [connect] (a1)
        (pi2) edge [connect] (a2)
        (pik) edge [connect] (ak)
        (a1) edge [connect] (pi2)
        (a2) edge [connect] (pik)
        (a1) edge [connect] (u1)
        (a2) edge [connect] (u2)
        (ak) edge [connect] (u3)
        ;
\end{tikzpicture}
\caption{Graphical representation of social learning filtering problem. The shaded nodes are known to us, while the white nodes are unknown}
\label{fig:graphsocial}
\end{figure}
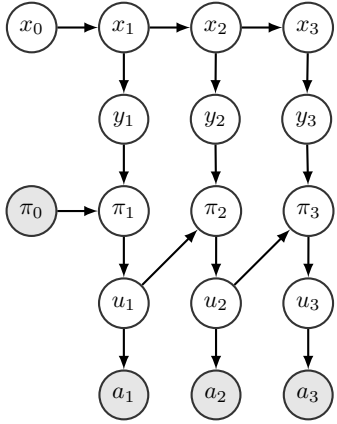

\subsubsection{Social Learning Protocol}
We start by describing the classic Bayesian  social learning protocol.
 In  social learning, the multiagent adversary learns (estimates) our state  based on their observations  of the state and  actions taken by previous agents.
Each agent takes an action $\laction_k$  once  in a predetermined sequential order indexed by $k=1,2,\ldots$   The index $k$ can also be viewed
as the  time instant when agent $k$ acts.  
Assume at the beginning of iteration $k$,
all agents have access to the public belief $\belief_{k-1}$ defined in  Step (iv) below.

The social learning protocol proceeds as follows
 \cite{BHW92,Cha04}:\\
 (i) {\em Private Observation}: At time $k$,
agent $k$  records a private observation $\obs_k\in \obspace= \{1,2,\ldots,\obsdim\}$
from the observation distribution $\oprob_{iy} = P(\obs_k=y|\state_k=i)$, $i \in \statespace$.
\\
(ii) 
{\em Private Belief}:  Using the public belief $\belief_{k-1} $ available at time $k-1$ (Step (iv) below), agent $k$   updates its private
posterior belief  $\priv_k(i)  =  \prob(\state_k = i| \laction_1,\ldots,\laction_{k-1},\obs_k)$ using the classical Bayesian filter (\ref{eq:belief}) which reads:
\begin{align}  \label{eq:hmm} \priv_k &= 
\frac{\oprob_{\obs_k} \tp^\p\, \belief}{ \ones^\p \oprob_y  \belief}, \;  
\oprob_{\obs_k} = \text{diag}(P(\obs_k|\state_k=i), i\in \statespace) . 
 \end{align}
 Here  $\priv_k$ is an $\statedim$-dimensional probability mass function (pmf). \\
  (iii)   {\em Myopic Action}: Agent  $k$  takes  action $$\laction_k\in \actionspace = \{1,2,\ldots, \actiondim\}$$ to  minimize its expected cost 
\beq
  \laction_k =  \arg\min_{\laction \in \actionspace} \E\{\cost(\state_k,\laction)|\laction_{1:k-1},\obs_k\}   =\arg\min_{\laction\in \actionspace} \{\cost_\laction^\p\priv_k\}.    \label{eq:myopic}
 \eeq
  Here $\cost_\laction = (\cost(i,\laction), i \in \statespace)$ denotes an $\statedim$ dimensional cost vector, and $\cost(i,\laction)$ denotes the cost  incurred when the underlying state is $i$ and the  agent chooses action $\laction$.\\
(iv) {\em Social Learning Filter}:   
Given the action $\laction_k$ of agent $k$,  and the public belief $\belief_{k-1}$, each  subsequent agent $k' > k$ 
performs social learning to
update the public belief $\belief_k$ according\footnote{For the reader unfamiliar with social learning, the remarkable aspect of the social learning filter is that the likelihood $\aoprob_{\laction}^\belief$ is an explicit function of the prior $\belief$; whereas in classical filtering the prior and likelihood are functionally independent. This dependence of the likelihood on the prior that causes social learning to have  unusual behavior such as herding and information cascades.} to the  ``social learning  filter":\ \beq \belief_k = \filter(\belief_{k-1},\laction_k), \text{ where } \filter(\belief,\laction) = 
 \frac{\aoprob_{\laction}^\belief\,\tp^\p\,\pi}{\filterd(\belief,\laction)} \label{eq:socialf}\eeq
where
$\filterd(\belief,\laction) = \ones^\p \aoprob^\belief_\laction \tp^\p \belief$ is the normalization factor of the Bayesian update.
In (\ref{eq:socialf}),  the public belief $\belief_k(i)  = P(\state_k = i|\laction_1,\ldots \laction_k)$ and $\aoprob^\belief_\laction  = \text{diag}(\prob(\laction|x=i,\belief),i\in \statespace) $ has elements
\begin{align*} 
 & P(\laction_k = \laction|x_k=i,\belief_{k-1}=\belief) = \sum_{\obs\in \obspace} \prob(\laction|y,\belief)\,\prob(y|\state_k=i) \\
 &   P(\laction_k=\laction|y,\belief) = \begin{cases}  1 \text{ if }  \cost_\laction^\p B_y \tp^\p \pi \leq \cost_{\tla}^\p \oprob_y\, \tp^\p\,\pi, \; \tla \in \actionspace \\
 0  \text{ otherwise. }  \end{cases}
  \end{align*}
where $I(\cdot)$ is the indicator function and $B_y$ are the classical observation likelihoods is defined in (\ref{eq:hmm}).

\subsubsection{Inverse Social Learning Filter}
Given  the above classical Bayesian social learning model,
our aim is to estimate the public belief $\belief_k$ given the state sequence $\state_{0:k}$ and
noisy measurements $\action_{1:k}$ of the agents actions $\laction_{1:k}$.

In order to proceed, first note that,
as  shown \cite{Kri12}, the social learning filter (\ref{eq:socialf}) has the following property:
The belief space $\Belief$ can be partitioned into   $\actiondim$ (possibly empty) subsets denoted $\Belief^1,\ldots,\Belief^\actiondim$ such
that the prior dependent likelihood $\aoprob^\belief_\action$ in
 (\ref{eq:socialf})
is a peiecewise constant function of prior belief $\belief$. i.e.,
$$\aoprob^\belief_\laction  = \aoprob^l_\laction , \quad \belief \in \Belief^l,
\quad l \in \{1,\ldots,\actiondim\}.$$

In complete analogy to Sec.\ref{sec:inversehmm}, we can now construct the inverse social learning filter: Let $\aprob_{\laction,\action} = \prob(\action|\laction)$.

For $k=1,2,\ldots$ construct $\Belief_k$   via the following recursion:
$$\Belief_k = \big\{ \filter(\belief,\laction), \;\laction \in \actionspace, \belief \in \Belief_{k-1} \big\},  \quad \Belief_0 = \{\belief_0\}.$$
 Note  $\Belief_k$ has
$\actiondim^k$ elements.
The inverse social learning filter computes the posterior $\post(\belief)$ and conditional mean of belief $\belief \in \Belief_{k+1}$ as follows:
\begin{align}
 & \post_{k+1}(\belief) = \frac{\aprob_{\laction^*,\action_{k+1}}
    \,  \sum_{\bbelief \in \Belief_k} \oprob_{\state_{k+1}, \obs_{\bbelief,\belief,\laction^*}}\, \post_k(\bbelief) }
  {\sum_{\belief \in \Belief_{k+1}} \aprob_{u^*,\action_{k+1}}
    \,  \sum_{\bbelief \in \Belief_k} \oprob_{\state_{k+1}, \obs_{\bbelief,\belief,u^*}}\, \post_k(\bbelief) }
\nonumber  \\
 & \E\{\belief_{k+1}|\action_{1:k+1},\state_{0:k+1}\} = \sum_{\belief \in \Belief_{k+1}} \belief \post_{k+1}(\belief)
 \label{eq:socialpost}
  \end{align}
Here $\laction^*(\bbelief,\belief)$   is the action such that $\belief = \filter(\bbelief,\laction)$. Also, $\obs(\bbelief,\belief,u)$ is the observation such that (\ref{eq:myopic}) holds.

\subsection{Inverse Bayesian Localization  using Conjugate Priors}
We conclude this section by discussing a  two-time scale inverse Bayesian localization problem. Recall that
localization refers to the case where our state is a random variable
$\state_0$ instead of a random process.

\subsubsection{Model}
Assume our state is a finite valued;  so $\state_0 \in \states$ is a random variable with prior $\belief_0$. (The  transition kernel $\tp = I$).
The adversary aims to localize (estimate)  our state based on  noisy measurements  that are exponentially distributed (this models the received radar signal power)
\beq \obs_k \sim \oprob_{\state_0, \obs} = \rate_{\state_0} \exp( - \rate_{\state_0} \obs), \quad k = 1,2,\ldots \label{eq:rxpower}
\eeq
In (\ref{eq:rxpower}), $\rateo \ole \rate_{\state_0}$ denotes the power gain parameter at the receiver when the true state is $\state_0$.
Note that  $k$ indexes the slow time scale. The  measurements  $\obs_k$ are obtained by integrating the received power over a fast time scale defined below.

The adversary computes belief $\belief_k$ using the Bayesian filter  (\ref{eq:belief}) with $\tp = I$ and observation kernel
$\oprob_{\state,\obs}$ with parameter $\rate = (\rate_1,\ldots,\rate_\statedim)^\p$. 
 Let $\bbelief_k(i) = \pdf(\state,\obs_{1:k})$ denote the un-normalized posterior density at time $k$. Since $\oprob_{\state\obs}$ is an exponential density, each element $\log \bbelief_k(i)$,
$i \in \states$
evolves with linear dynamics:
\beq
\log \bbelief_{k+1} (i) = \log \rate_i +  \log \bbelief_k(i)  - \rate_i \obs_{k+1}. \label{eq:localstate}
\eeq
Equivalently, relative to some fixed location 1, defining $\dbelief_k(i) = \log \bbelief_k(i) - \log \bbelief_k(1)$, the dynamics of the adversary's belief update are given by the linear stochastic system
\beq  \dbelief_{k+1}(i) = \dbelief_k(i)  +  \log \frac{\rate_i }{\rate_1} - (\rate_i - \rate_1) \obs_{k+1}
\label{eq:dbelief}
\eeq
Let $t=1,\ldots,\fastT$ denote the fast time scale at each epoch $k$.
We assume that our measurement of the adversary's radiated power in units of decibels (relative to location 1)  at fast time scale $(\fast,k)$ in each location $i \in \stateswi$  is
\beq  \action_{\fast,k}(i) =
\mean_i +  \frac{\onoise_{\fast,k}(i)}{\sqrt{|\dbelief_k(i)|} }
 \label{eq:localobs} \eeq
 where $\mean_i$ is the path loss and $\onoise_{\fast,k}(i)$ is zero mean unit variance iid Gaussian.
 Equivalently the absolute power is proportional to $ \exp(\mu_i) \, z_{\fast,k}(i)$ where $z_{\fast,k}(i) = (\belief_k(1)/\belief_k(i))^{\onoise_{\fast,k}(i)}$ is
 log-normal distributed.
The intuition behind  (\ref{eq:localobs}) is as follows: Since the adversary's radar radiates energy
to location $i$ proportional to $\belief_k(i)$,  so our variance of the measurement noise is inversely proportional
to $\belief_k(i)$. Hence the scaling of the noise variance  by $1/\Delta_k(i)$
in (\ref{eq:localobs}) in log-scale.

\subsubsection{Inverse Localization}
The aim is to estimate adversary's belief $\dbelief_k$ given $\action_{1:\fastT,k}$, i.e., compute $\pdf(\dbelief_k|\action_{1:\fastT,k})$
given $\state_0$.
Recall $\action_{1:\fastT,k}=(\action_{1,k},\ldots,\action_{\fastT,k})$.
 From (\ref{eq:localstate})
 $$ \dbelief_k(i) = k \log \frac{\rate_i}{\rate_1}+  \dbelief_0(i)
  - (\rate_i - \rate_1)\,
\sum_{\dtime=1}^k \obs_\dtime 
 $$
 
Denote $S_k = \sum_{\dtime=1}^k \obs_\dtime $.
Since $\{\obs_\dtime\}$ are iid exponential distributed, hence
 $S_k$  has a  Gamma distribution
with pdf  
$$ \pdf_{S_k}(s) = \Ga(k,\rateo) = \rateo \,e^{-\rateo s} \,\frac{ (\rateo s)^{k-1}}{(k-1)!},  \quad s \geq 0
$$ 
The Gaussian likelihood (\ref{eq:localobs}) with Gamma prior for the precision (inverse of variance)  form a Bayesian conjugate pair; implying that the posterior is  also a Gamma  distribution. Therefore, we have
the following Bayesian estimator of the adversary's relative belief $\dbelief_k$:

\begin{thm} For the model (\ref{eq:rxpower}), (\ref{eq:localstate}), (\ref{eq:localobs}), our posterior of the adversary's belief given $\action_{1:\fastT,k},\state_0$  has a Gamma density:
$$ \pdf(\log \dbelief_k(i)= s| \action_{1:\fastT,k} ,\state_0) = \frac{1}{\rate_i-\rate_1} \pdf_{S_k| \action_{1:\fastT,k}}\left( \frac{\parn_i - s}{\rate_i-\rate_1} \right)  $$
where $s\geq 0, \parn_i = k \log \rate_i/\rate_1 + \log \dbelief_0(i)$ and
  $$ \pdf_{S_k| \action_{1:\fastT,k}}(s | \action_{1:\fastT,k}) = \Ga\big(k+ \frac{\fastT}{2}, \rateo + \frac{1}{2} \big( \sum_{\fast=1}^{\fastT} \action_{\fast,k}(i) -
  \mean_i \big)^2 \big)$$
  \end{thm}

\section{Sequential Localization Game: Covariance Bounds}

Suppose our state $\state_0$ is a random variable and the adversary localizes our state via a Bayesian estimator. We observe the adversary's actions  and estimate the adversary's belief. Since the adversary's belief $\belief_k$ converges to a Dirac mass centered on our  state $\state_0$ with probability one (under suitable regularity conditions\footnote{The Bernstein von-Mises theorem \cite{Vaa00}  yields  a central limit theorem for the Bayesian estimator and therefore the asymptotic convergence rate.}), our posterior
$\alpha_k(\belief)$ computed via the inverse filter, will also converge with probability one to a Dirac mass centered on $\state_0$. The question we address in this section is:
{\em How much slower is the convergence of our inverse filter compared to the adversary's Bayesian filter?}

To make our analysis illustrative,
throughout this section, we 
consider a scalar-valued localization problem with  the inverse Kalman filter.
Assume $\state_0 \in \reals$ (scalar state),
  $\statem = 1$, $\obsm=1$,  $\snoisecov = 0$, $\onoisecov = 1$; so the adversary observes us with SNR  1.
Assume the adversary's prior is non-informative; so $\hstate_0 = 0_{\statedim},\kalmancov_0 = \infty$.
\subsection{Localization and Inverse Kalman Filter}
We start with an elementary but useful observation.
Using the information filter (\ref{eq:info}), it readily follows that the adversary's localization estimate is
$$ \hstate_{k+1} = \frac{1}{k+1} \sum_{\dtime=1}^{k+1} \obs_\dtime
,  \qquad \kalmancov_{k+1} = \frac{1}{k+1}$$
Supose we observe the adversary's action with $\enemySNR_k = \enemyobsm_k^2/ \enemyonoisecov_k$ and initial condition $\enemykalmancov_0 = \infty$. Then  the covariance of our estimate (using inverse Kalman filter (\ref{eq:inversekf})) of the adversary's belief evolves as
 \beq \enemykalmancov_{k+1} = \frac{k^2 \enemykalmancov_k + 1  }{(k+1)^2 + \enemySNR_{k+1}\,(k^2 \enemykalmancov_k + 1)}  \label{eq:case1} \eeq
From (\ref{eq:case1}), it is easily seen that $$\enemykalmancov_k \vert_{\enemySNR= 1} <
\enemykalmancov_k\vert_{\enemySNR = \Sigma_k} < 
 \kalmancov_k$$ i.e., our covariance of the adversary's state estimate   is strictly smaller than the adversary's covariance of our state.

 \subsection{Sequential Localization Game in Adversarial Setting}
 Thus  far  we assumed that  $\state_0$ is known to us and our aim was to estimate the adversary's estimate. In this section, our framework is different in two ways. First, we assume that
 neither us nor the adversary know the underlying state $\state_0$. For example, $\state_0$ could denote the threat level of an autonomous unidentified drone hovering relative to a specific location. Second, there is feedback; the adversary's action affects our estimate and our action affects the adversary's estimate.

 The setup is naturally formulated in game-theoretic terms.
 The  adversary and us play a sequential game  to localize  (estimate)~$\state_0$:
 \begin{compactenum}
 \item At odd time instants $k=1,3,\ldots$, the  adversary takes active  measurements of the target  and makes decisions based on its estimate.  We  use measurements of the adversary's decisions to localize the target. The details are as follows: The
 adversary  has a noisy measurement   $\action_{k-1}= \laction_{k-1}+ \anoise_{k-1}$ where 
 $\laction_{k-1} = \hstate_{k-1}$ is our action. The adversary assumes that our state estimate is $\action_{k-1}$ and
its Kalman filter  (\ref{eq:kalman}) then  yields
the updated state  estimate
$$\hstate_k = (1 - \kg_{k}) \action_{k-1} + \kg_{k} \obs_k, \qquad
\text{ where } \kg_k = \kalmancov_k.
$$
The adversary then 
 takes action $\laction_k= \hstate_k$ to myopically minimize the mean square error of the estimate.
We eavesdrop (observe)  the adversary's action in noise as $\action_k=\hstate_k + \anoise_k$. Assume $\anoise_k \sim \normal(0,\anoisecov)$ iid. So
$$ \action_{k} = (1 - \kg_{k}) \action_{k-1} + \kg_{k} \onoise_{k} +
\kg_{k} \state_0 + \anoise_{k} $$
Since  $\action_{k-1} = \laction_{k-1}+ \anoise_{k-1}$,
 our effective observation equation of $\state_0$  based on our measurement of the  adversary's action $\action_k$ is
\beq \eobs_k \ole  \frac{\action_k }{\kg_k} - \frac{1-\kg_k}{\kg_k}\,\laction_{k-1} 
= \frac{1-\kg_k}{\kg_k}\,\anoise_{k-1}+ \onoise_k + \frac{\anoise_k}{\kg_k} + \state_0
\label{eq:eobs}
\eeq
We know  all the quantities in  the middle  equation; $\laction_{k-1}$ was our action taken at time $k-1$, and $\action_k$ is our measurement of the adversary's action at time $k$.

To summarize, at odd time instants,
 $\eobs_k = \frac{1-\kg_k}{\kg_k}\,\anoise_{k-1}+ \onoise_k + \frac{\anoise_k}{\kg_k} + \state_0$ is our effective observation of the unknown state $\state_0$ purely based on sensing the adversary's actions.
\item At even time instants $k=2,4,\ldots,$ we take active measurements of the target and take actions based on our estimates. The adversary uses measurements of our actions to localize the target. We  use our measurement $\action_{k-1}$ of the adversary's action, obtain a noisy observation $\obs_k$ of $\state_0$
 and then  compute belief $\belief_k$. We then take action $\laction_k$ to myopically minimize the mean square error of the estimate. The adversary eavesdrops on (observes) our action in noise as $\action_k$.
\end{compactenum}
The above setup constitutes a social learning sequential game \cite{MMST18}. Since the decisions are made myopically each time (to minimize the mean square  error), the strategy profile at each time constitutes a Nash equlibrium. More importantly \cite[Theorem 5]{MMST18}, the asymptotic behavior (in time) is captured by a social learning equlibrium (SLE).

\begin{thm} \label{thm:gamecov}
  Consider the sequential game formulation for localizing the random variable state $\state_0$. Then for large $k$,
  the covariance of our estimate of the state $\state_0$ is
  $\kalmancov_k = 2 k^{-1}$ (which is twice the covariance for classical localization).
  The same result holds for the adversary's estimate.
\end{thm}

{\em Remark}. Theorem~\ref{thm:gamecov} has two interesting consequences. \begin{compactenum}
  \item The asymptotic covariance is independent of $\anoisecov$ (i.e.  SNR of  our observation of the adversary's action)
as long as $\anoisecov > 0$. 
If $\anoisecov = 0$, there is a phase change and the covariance $\kalmancov_k = k^{-1}$
as in the standard Kalman filter.
\item 
  In the special case where  only (\ref{eq:eobs}) is observed at each time
(and the $\anoise_{k-1}$ term is omitted),
  the
 formulation reduces to the Bayesian social learning problem 
  of \cite[Chapter 3]{Cha04}. In that case, \cite[Proposition 3.1]{Cha04} shows that  $\kalmancov_k = O(k^{-1/3})$. 
\end{compactenum}
To summarize, Theorem \ref{thm:gamecov} confirms the intuition  that sequentially learning the state indirectly from the adversary's actions (and the adversary learning the state from our actions) slows down the convergence of the localization estimator.

\section{Maximum Likelihood Parameter Estimation of Adversary's Sensor}\label{sec:mle}
So far we have discussed Bayesian estimation of  the adversary's belief state.  In comparison, this section considers parameter estimation.
Our aim is to  estimate the adversary's observation kernel  $\oprob$ in (\ref{eq:model}) which  quantifies the accuracy of  the adversary's sensors.
We assume that  $\oprob$   is parametrized by an $\modeldim$-dimensional vector $\model \in \Model$ where
$\Model$ is a compact subset of $\reals^\modeldim$. Denote the parametrized
observation kernel as $\oprob^\model$.
Assume  that both us and the adversary know $\tp$ (state transition kernel\footnote{As mentioned earlier, otherwise the adversary estimates $\tp$ as $\hat{\tp}$ and we need to estimate the adversary's estimate as
  $\hat{\hat{\tp}}$. This makes the estimation task substantially more complex.}) and $\aprob$ (probabilistic map from adversary's belief to its action).
Then given our state sequence $\state_{0:\horizon}$  and adversary's action sequence $\action_{1:\horizon}$, our aim is to compute the maximum likelihood estimate (MLE) of $\model$. That is, with $\lik(\model)$ denoting the log likelihood, compute
$$ \mle = \argmax_{\model \in \Model} \lik(\model), \quad \lik(\model)= \log \pdf(\state_{0:\horizon},\action_{1:\horizon} | \model) .$$

\subsection{General Purpose Optimization for MLE}
The likelihood can be evaluated from
 the un-normalized version of the inverse filtering recursion (\ref{eq:post}) which reads
\beq
\unpostm_{k+1}(\belief) = \aprob_{\belief,\action_{k+1}}
    \,  \int_\Belief \oprob^\model_{\state_{k+1}, \obs^\model_{\belief_k,\belief}}\, \unpostm_k(\belief_k) d\belief_k
\label{eq:unpost}
\eeq
Given the un-normalized  filtering recursion (\ref{eq:unpost}), the log likelihood is  the normalization term at time $\horizon$:
\beq  \lik(\model) = \log \int_\Belief \unpostm_\horizon(\belief) d\belief. \label{eq:loglik} \eeq

Given  (\ref{eq:loglik}),
the MLE can be computed using  a general purpose numerical optimization algorithm such as  {\tt fmincon}  in Matlab. 
This uses the interior reflective Newton method \cite{CL96} for large scale problems and sequential quadratic programming 
\cite{Ber00} for medium scale problems. 
  In general $\lik(\model)$ is non-concave in $\model$ and the constraint set $\Model$ is non-convex. So the algorithm
  at best will converge to a  local 
  stationary point of the likelihood function.
 General purpose optimizers such as 
  {\tt fmincon} in Matlab  allow for the gradient  $\nablam \lik(\model)$ and Hessian
$\nablam^2 \lik(\model)$ of the likelihood  as inputs to  the algorithm.
These are typically computed via the so called sensitivity equations. For the general inverse filtering problem, the sensitivity equations are formidable.

{\em Remark. EM Algorithm}: The Expectation Maximization (EM) algorithm is a popular numerical algorithm for computing the MLE especially when the Maximization (M) step can be computed in closed form. It is widely used for computing the MLE of the parameters of a linear Gaussian state space model \cite{SS82,EK99}. Unfortunately, for estimating the adversary's observation model,  due to the time evolving dynamics in (\ref{eq:inversekfparam}), the EM algorithm is not useful
since the M-step involves a non-convex optimization that cannot be done in closed form. There is no obvious way of choosing the latent variables to avoid this non-convexity.

\subsection{Estimating Adversary's Gain Matrix  in Linear Gaussian Case}
 Consider the setup  in Sec.\ref{sec:inversekalman} where our dynamics are linear Gaussian and the adversary observes our state linearly in Gaussian noise (\ref{eq:lineargaussian}). The adversary
estimates our state using a Kalman filter, and we estimate the adversary's estimate using the inverse Kalman filter (\ref{eq:inversekf}).
Using (\ref{eq:loglik}), (\ref{eq:inversekf}), (\ref{eq:inversekfparam}),  the log likelihood  for the adversary's observation  gain matrix $\param=\hobsm$  based on our measurements is
\begin{align}
  \lik(\model) &= \text{const} - \frac{1}{2} \sum_{k=1}^\finaltime \log | \enemySig^\param_k| - \frac{1}{2} \sum_{k=1}^\finaltime
  \innovations_k^\p\, (\enemySig_k^\param)^{-1} \,\innovations_k\nonumber \\
  \innovations_k &= \action_k - \enemyobsm_k^\param\, \enemystatem_{k-1}^\param \enemystate_{k-1} - \enemyinputm_{k-1}^\param \state_{k-1} \label{eq:inversekflik}
\end{align}
where $\innovations_k$ are the innovations of the inverse Kalman filter (\ref{eq:inversekfequations}).
In (\ref{eq:inversekflik}), our state $\state_{k-1}$ is known to us and therefore is a known exogenous input. Also note from (\ref{eq:inversekfparam}) that $\enemystatem_{k}, \enemyinputm_k$
are explicit functions of $\obsm$, while $\enemyobsm_k$ and $\enemysnoisecov_k$ depend on $\obsm$ via the adversary's Kalman filter.

The following known result  summarizes the consistency  of the MLE for the for the adversary's gain matrix $\obsm$ in the linear Gaussian case.
\begin{thm}
  Assume that $\statem$ is stable, and the state space model (\ref{eq:lineargaussian})
is an identifiable minimal realization (which implies controllability and observability). Then the adversary's Kalman filter variables 
$\Sig_k,\kg_k,\kalmancov_k$ converge to steady state values geometrically fast in  $k$  \cite{AM79} implying that asymptotically the system (\ref{eq:inversekf}) is stable linear time invariant. Also, the  MLE $\mle$ (maximizer for (\ref{eq:inversekflik}))  for the adversary's observation matrix $\obsm$ using (\ref{eq:inversekf})  is unique and a strongly consistent estimator \cite{Cai88}.
\end{thm}

{\bf Numerical Examples}.
To provide insight,  Figure \ref{fig:kf} displays  the log-likelihood versus adversary's gain matrix $\hobsm$  using (\ref{eq:inversekflik}) in the scalar case.
The parameters in the simulation are $\statem = 0.4$, $\snoisecov = 2$,
$\onoisecov = 1$, $ \anoisecov= 1$ in (\ref{eq:linearaction}) with $\finaltime = 1000$. The four sub-figures correspond to true values of $\obsmt = 0.5,1,5,2,3$ respectively.

Each sub-figure in  Figure \ref{fig:kf} has two plots. The plot in red is  the  log-likelihood of $\hobsm \in (0,10]$ evaluated based on the adversary's observations using the  standard  Kalman filter. (This is the classical log-likelihood of the observation gain of a Gaussian state space model.)  The plot in blue is the log-likelihood of $\hobsm\in (0,10]$ computed using (\ref{eq:inversekflik}) based on  our measurements of the adversary's action using the  inverse Kalman filter (where the adversary first estimates our state using a Kalman filter) - we call this the inverse case.

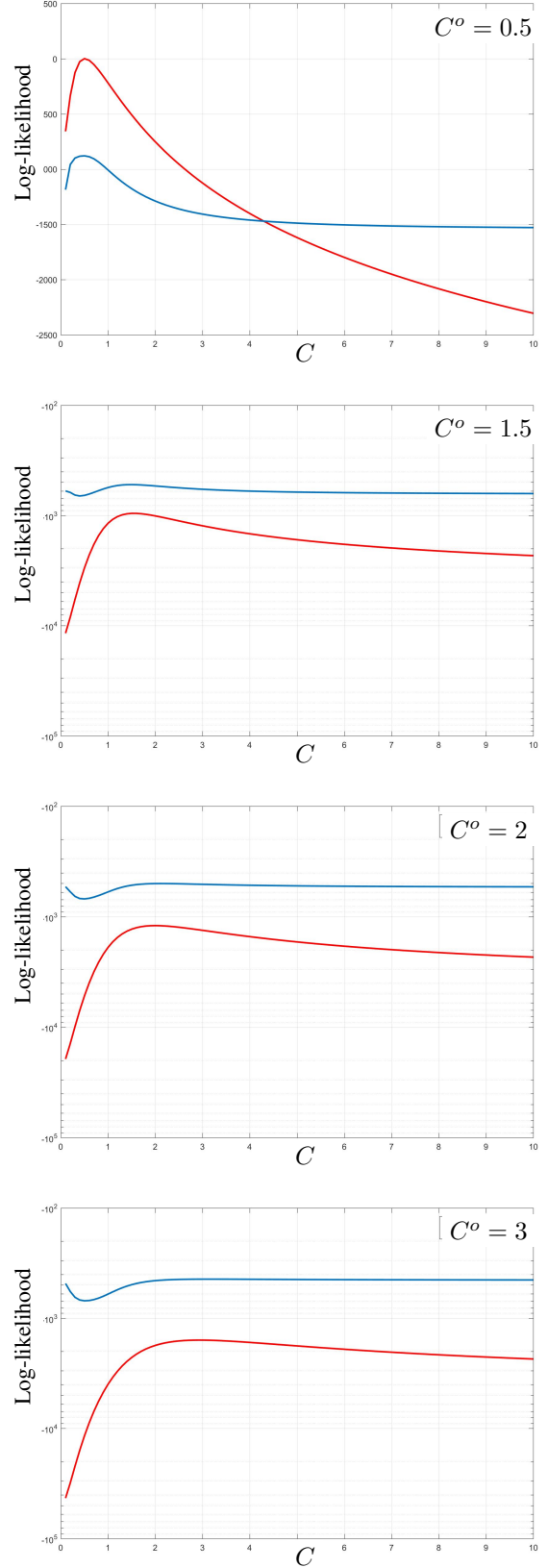
\begin{figure}
  \begin{overpic}[scale=0.12,unit=1mm]{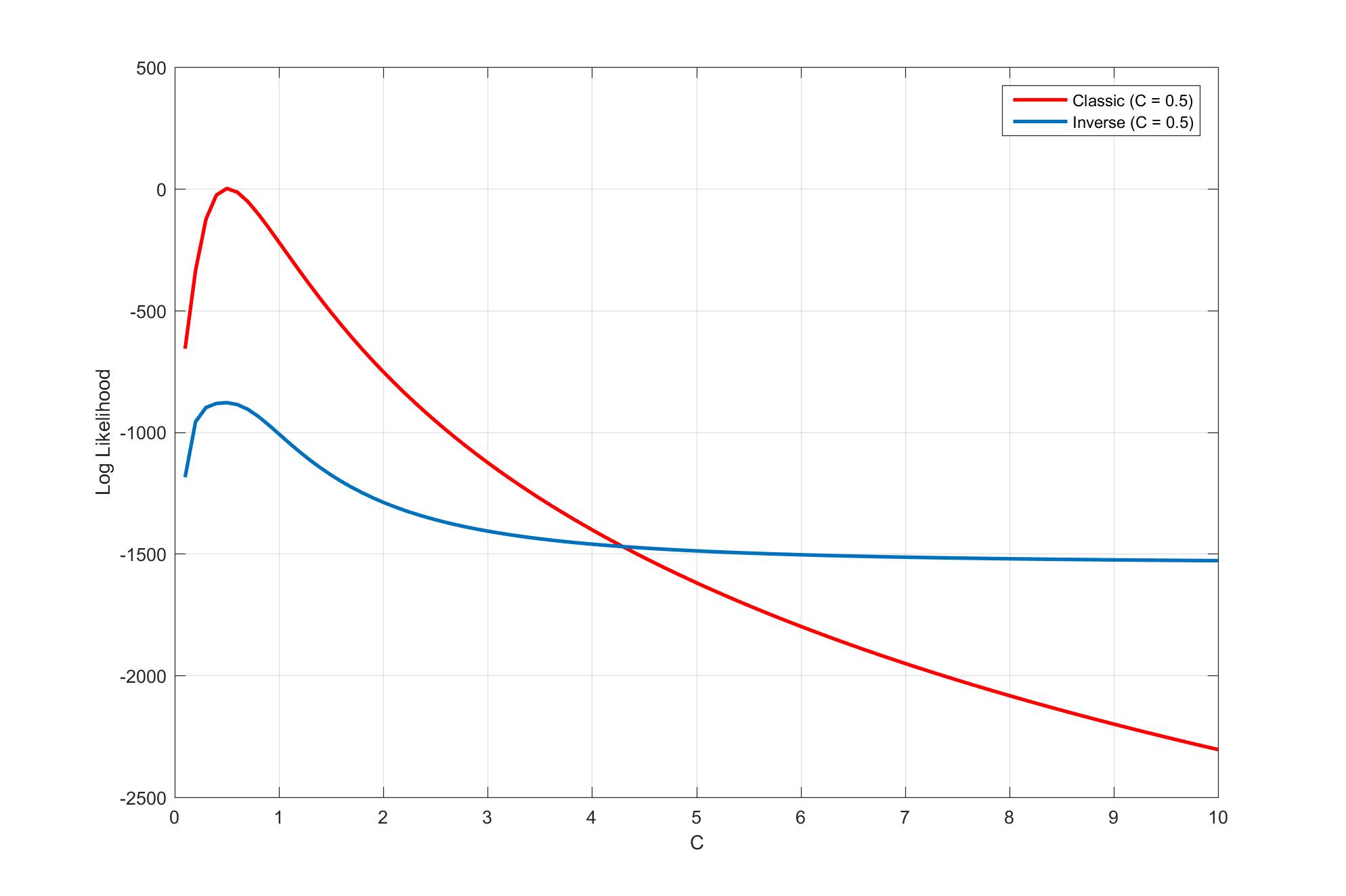}
    \put(50.15,2.9){\colorbox{white}{$\hobsm$}}
    \put(4,28){\rotatebox{90}{\colorbox{white}{\small Log-likelihood}}}
      \put(73,56.5){\colorbox{white}{$\obsmt=0.5$}}
  \end{overpic}
  \begin{overpic}[scale=0.12,unit=1mm]{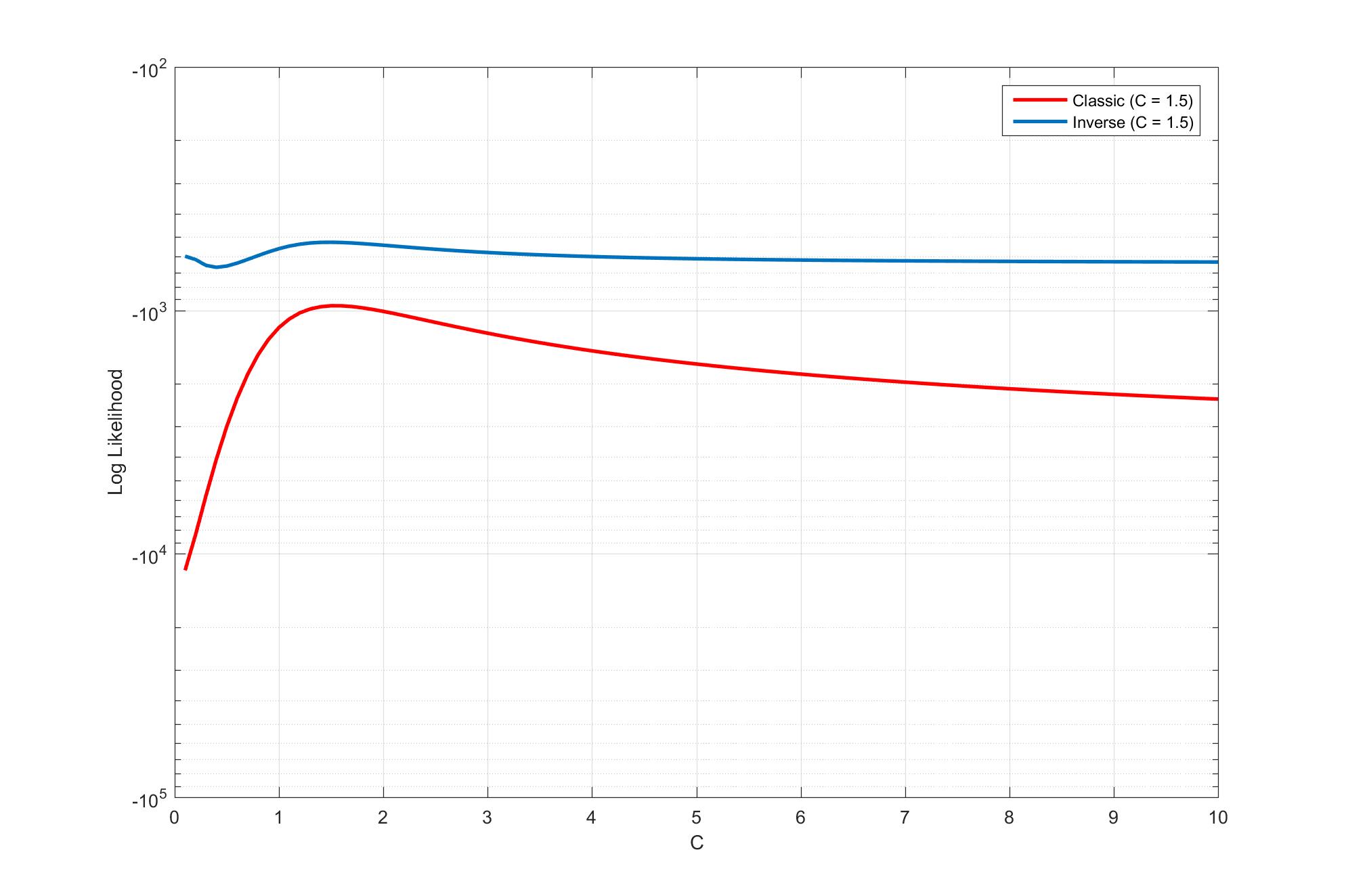}
    \put(50.15,2.9){\colorbox{white}{$\hobsm$}}
    \put(4,28){\rotatebox{90}{\colorbox{white}{\small Log-likelihood}}}
     \put(73,56.5){\colorbox{white}{$\obsmt=1.5$}}
  \end{overpic}
  \begin{overpic}[scale=0.12,unit=1mm]{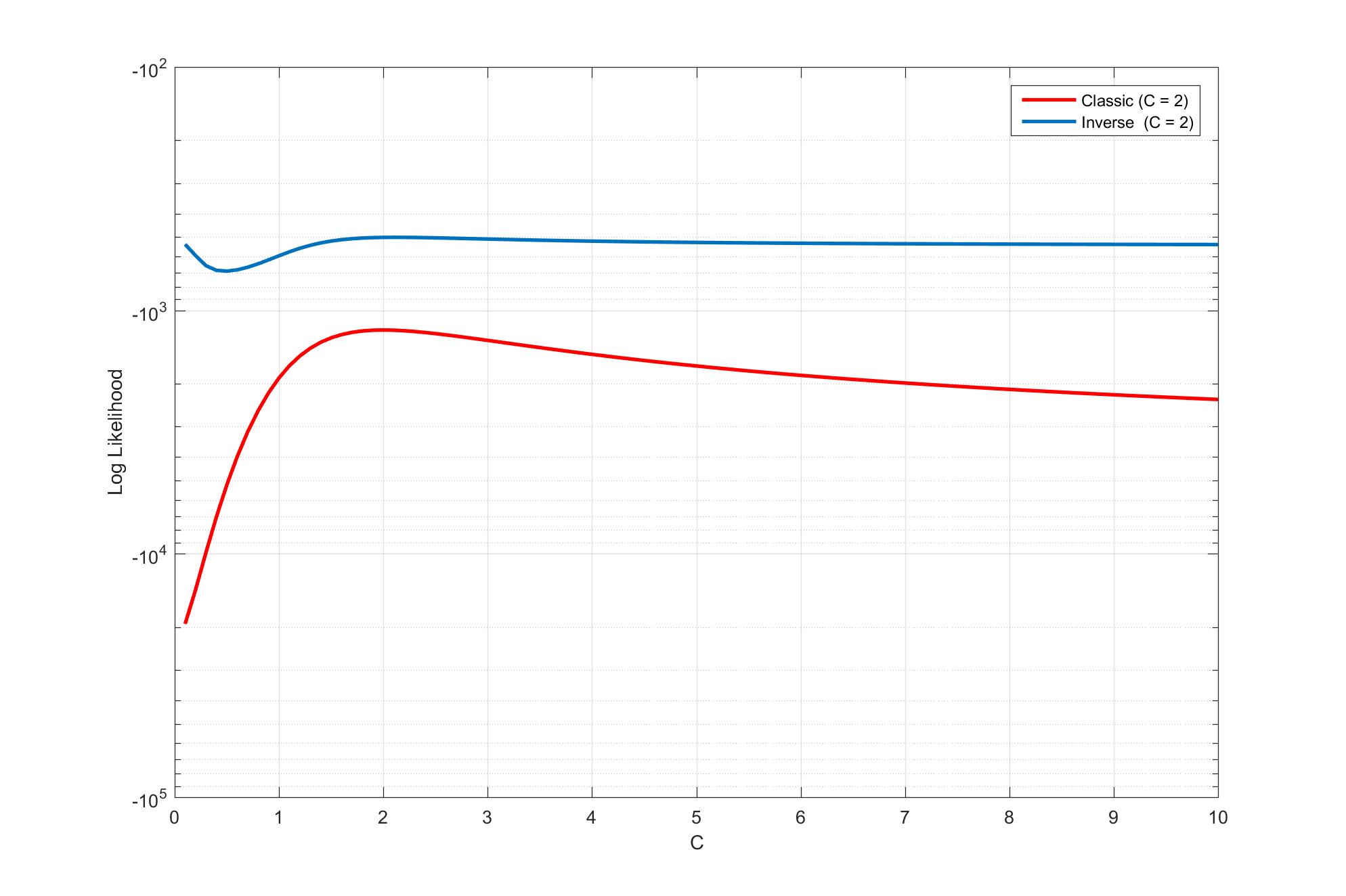}
    \put(50.15,2.9){\colorbox{white}{$\hobsm$}}
    \put(4,28){\rotatebox{90}{\colorbox{white}{\small Log-likelihood}}}
     \put(75.5,56.5){\colorbox{white}{$\obsmt=2$}}
  \end{overpic}
  \begin{overpic}[scale=0.12,unit=1mm]{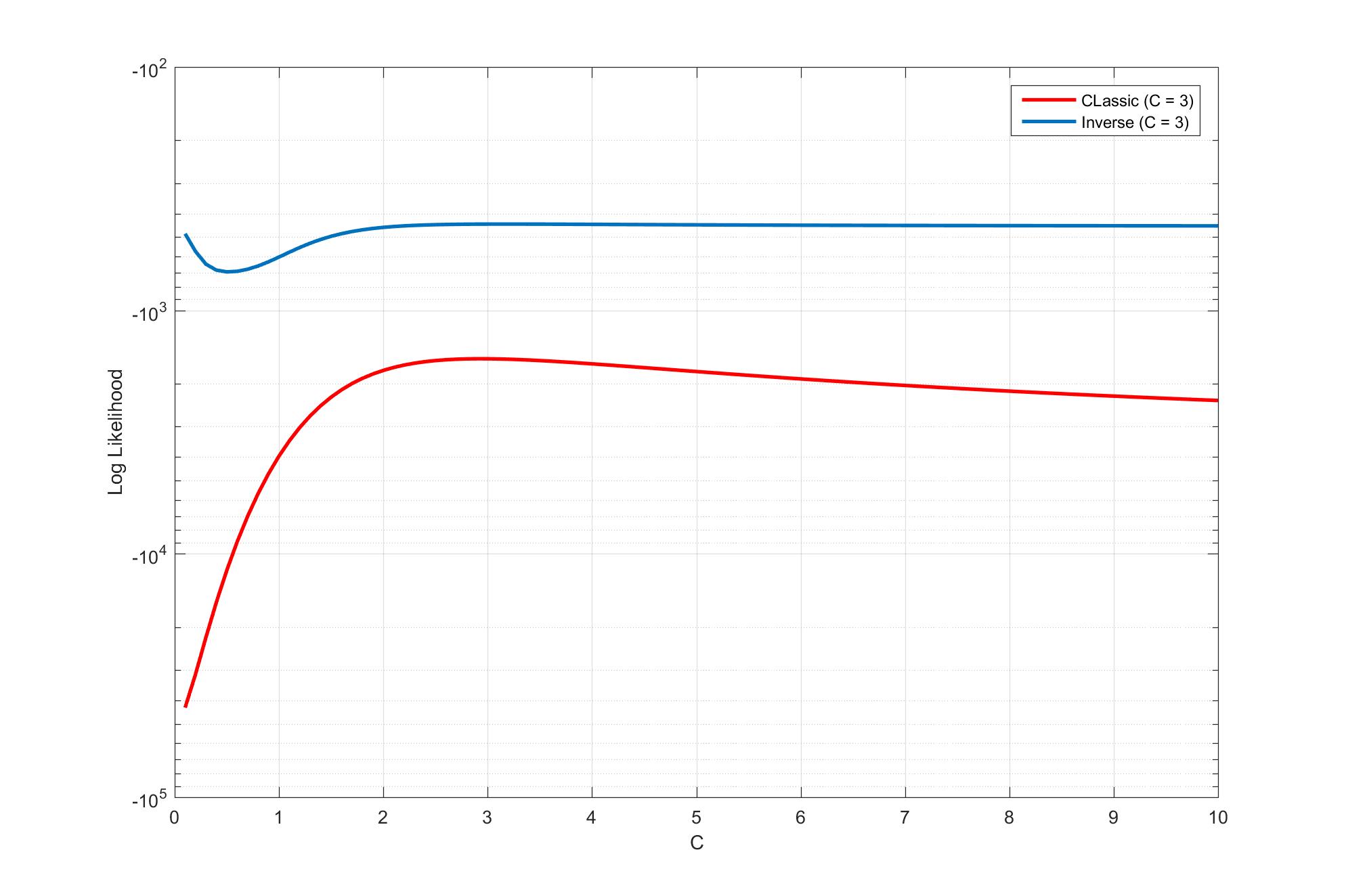}
    \put(50.15,2.9){\colorbox{white}{$\hobsm$}}
    \put(4,28){\rotatebox{90}{\colorbox{white}{\small Log-likelihood}}}
     \put(75.5,56.5){\colorbox{white}{$\obsmt=3$}}
  \end{overpic}
\caption{Log-Likelihood as a function of enemy's gain $\hobsm\in (0,10]$ when true value is $\obsmt$. The red curves denote the log-likelihood of $\hobsm$ given the enemy's measurements. The blue curves denotes the log-likelihood of $\hobsm$ using the inverse Kalman filter given our observations of the enemy's action. The plots show that it is more difficult to compute the MLE for the inverse filtering problem due to the almost flat likelihood (blue curves) compared to red curves.}
\label{fig:kf}
\end{figure}
 Figure \ref{fig:kf} shows  that the log likelihood in the inverse case (blue plots) has a less pronounced maximum compared to the standard case (red plots). This implies that numerical algorithms for computing the MLE of the  enemy's gain using our observations of the adversary's actions (via the inverse Kalman filter) will converge much more slowly than the classical MLE (which uses  the adversary's observations). This is intuitive  -  our estimate of the adversary's parameter is based on the adversary's estimate of our state - so there is more noise to contend with.

{\bf Cramer Rao (CR) bounds}. Is it instructive to compare the   CR bounds for MLE of $\obsm$ for the classic model versus that of the inverse Kalman filters model. Table \ref{tab:cr} displays the CR bounds (reciprocal of expected Fisher information) for the four examples considered above evaluated using via the algorithm in  \cite{CS96}. 
It shows that the covariance lower bound for the inverse case is substantially higher than that for the classic case. This is again
consistent with the intuition that estimating the adversary's parameter based on its actions (which is based on its estimate of us) is more difficult than directly estimating $\obsm$ in a classical state space model.

\begin{table} \centering
\begin{tabular}{|c|c|c|}
  \hline
  $\obsmt$ & Classic & Inverse \\
  \hline
  0.5 & $0.24 \times 10^{-3} $ & $5.3 \times 10^{-3}$ \\
  1.5 & $1.2 \times 10^{-3}$ &   $37 \times 10^{-3}$ \\
  2    & $2.1 \times 10^{-3} $ &  $70 \times 10^{-3}$ \\
  3  & $ 4.6 \times 10^{-3}$ & $ 336 \times 10^{-3}$\\
  \hline                       
\end{tabular}
\caption{Comparison of Cramer Rao bounds for $\obsm$ - classical state space model vs inverse Kalman filter model}
\label{tab:cr}
\end{table}

\section{Optimal Probing of Adversary}
Thus far we have  discussed estimating the adversary's  belief state and observation  matrix. This section deals with optimal probing also known  as input design.
 Suppose  our probe signal  $\state_k$ is a finite state Markov chain with transition matrix
 $\tp$. Recall from (\ref{eq:model}) that  the adversary observes $\state_k$ via observation probabilities  $\oprob$ and deploys a Bayesian filter (involving $\tp,\oprob$) to compute its belief  $\belief_k$, then chooses action $\laction_k$ that we observe as $\action_k$. This section addresses the question:  {\em How should we choose the transition matrix  $\tp$ of our probe signal $\state_k$  to minimize the covariance
 of our estimate of the adversary's observation kernel $\oprob$?}

Below we consider two approaches. The first  approach is analytical -- using stochastic dominance, we establish a partial ordering amongst probe transition  matrices which results in a  corresponding ordering of our SNR of the adversary's actions. 
The second approach is  numerical -- we describe a stochastic gradient algorithm to optimize the probe transition matrix.

\subsection{Stochastic Dominance Ordering of Optimal Probing Transition Matrices} In this subsection, we construct a partial ordering for transition  matrices $\tp$ of our probe signal $\state_k$ which results in a corresponding  ordering of our SNR of the adversary's action. The result is useful since  computing the optimal $\tp$ is non-trivial (see next subsection), yet we can  compare  the SNRs to two different probing signals   without brute force computations.  Note that maximizing this SNR
 can be viewed as a surrogate for minimizing the covariance of our estimate of the adversary's observation kernel~$\oprob$.

 \subsubsection{Assumptions}
We  assume (these assumptions are discussed below)
\begin{compactenum}
\item[(A1)] $\tp$ is totally positive of order 2, see \cite{KR80,Kri16}. That is every second order minor of $\tp$  is non-negative.
\item[(A2)] The adversary chooses action $\laction_k = \fun(\levels^\p \belief_k) $ where $\levels $  is a  vector with increasing elements and $\fun$ is any increasing non-negative function (e.g.\  quantizer).
 Recall the adversary observes $\state_k$ as $\obs_k$ and computes its Bayesian belief $\belief_k$ via~(\ref{eq:belief}) using the transition matrix $\tp$.
\end{compactenum}

  Given our observation $\action_k = \laction_k +\anoise_k$ of the adversary's action $\laction_k$, define the signal to noise ratio
  $$\enemySNR(\tp) = \frac{\E\{\laction_k^2\}}{\anoisecov}.$$

  \subsubsection{Copositive  Dominance and Monotone Likelihood  Ratio Dominance}
  In order to compare  $\enemySNR(\tp)$ for different probing  transition matrices $\tp$, we introduce two definitions: copositive dominance of transition  matrices and monotone likelihood ratio dominance of beliefs.

Copositivity generalizes positive semi-definiteness. Recall a matrix $\copomat$ is positive semi-definite if $\belief^\p \copomat \belief \geq 0$ for  all $\belief \in \reals^{\statedim}$. In comparison, a matrix $\copomat$ is copositive if $\belief^\p \copomat \belief \geq 0$ for all $\statedim$-dimensional probability  vectors $\belief$.

  Given two transition matrices $\tpi$ and $\tpii$, define the matrices
  $\copomat^j$, $j = 1,\ldots
\statedim-1$ each of dimension $\statedim \times \statedim$ as
\beq
\begin{split}
   \copomat^{j} &= \cfrac{1}{2}\left[\gamma^{j}_{mn} + \gamma^{j}_{nm}\right]_{\statedim\times\statedim},
\\
\gamma^{j}_{mn} &=  \tp_{m,j}(1)\tp_{n,j+1}(2) - \tp_{m,j+1}(1)\tp_{n,j}(2) .
\end{split} \eeq

\begin{definition}[Copositive Ordering of Transition Matrices \cite{Kri16}]  \label{def:lR}
Given two  transition matrices $\tpi$ and $\tpii$, we say that   $\tpi$ is  copositively dominated by  $\tpii$, denoted as 
$ \tpi \lR \tpii  $ if
  each matrix  $\copomat^{j,\action}$, $j=1\,\ldots,\statedim-1$  is copositive, i.e., 
$$
  \belief^\p  \copomat^{j,\action} \belief  \geq 0 ,  \quad \forall \belief \in \Belief .
                                             $$
\end{definition}

Next we define the monotone likelihood ratio partial order for probability  mass functions.

\begin{definition}[Monotone Likelihood Ratio (MLR) Dominance]  \label{def:mlr}
Let $\belief_1, \belief_2 \in \Belief$ be any two probability vectors.
Then $\belief_1$ is greater than $\belief_2$ with respect to the MLR ordering -- denoted as
$\belief_1 \gr \belief_2$ -- if
\beq \belief_1(i) \belief_2(j) \leq \belief_2(i) \belief_1(j), \quad i < j, \quad i,j\in \{1,\ldots,\statedim\}.
\label{eq:mlrorder}\eeq 
\end{definition}
Similarly $\belief_1 \lr \belief_2$ if  $\leq$ in (\ref{eq:mlrorder}) is replaced
by a $\geq$. \\
The MLR stochastic order is  useful  in a Bayesian context since it is closed  under conditional expectations.
That is, $X\gr Y$  implies $\E\{X|\F\} \gr \E\{Y|\F\}$ for any two random variables $X,Y$ and sigma-algebra $\F$
\cite{MS02}.

{\em Discussion of Assumptions}: (A1) is widely used in Markov chain structural analysis; see \cite{Kri16} for specific examples and also the classic paper \cite{KR80}. (A2) says that  larger  beliefs  (wrt MLR order)
of the adversary correspond to larger actions taken by the adversary. This implies that
$\laction$ is increasing with our state $\state$. For example, if $\state$ denote the threat levels perceived by the adversary, then the larger our state, the larger the adversary's action.

\subsubsection{Main Result}
With the above definitions, we are now ready to state the main result of this subsection.

As  shown in \cite{Kri16}, copositive dominance is a sufficient condition for the one step Bayesian  belief  updates (\ref{eq:belief})  using $\tpi$ and $\tpii$, respectively,  to satisfy $\filter(\belief,\obs;\tpi) \lr \filter(\belief,\obs,\tpii)$ where $\lr$ denotes likelihood  ratio dominance. Assumption (A1) globalizes this assertion and implies that adversary's beliefs for transition matrices $\tpi$ and $\tpii$ satisfy $\belieftpi_k \lr \belieftpii_k$ for all time $k$. This implies that $\E\{\laction_k^2\}$ using $\tpi$ is smaller than that using $\tpii$. 
We summarize this argument as the following  result.

\begin{thm} \label{thm:snr}
  Suppose transition matrices $\tpi$, $\tpii$ satisfy (A1), (A2)   and $\tpi \lR  \tpii$ (copositive ordering). Then
  \begin{compactenum}
    \item The belief updates using transition matrices $\tpi$ and $\tpii$ satisfy: $\belieftpi_k \lr \belieftpii_k$ for all time $k$.
    \item 
      Therefore, our probe signal $\{\state_k\}$  generated with transition matrix $\tpi$ incurs a smaller signal to noise ratio compared  to probe signal with transition matrix $\tpii$, that is,
        $\enemySNR(\tpi) \leq \enemySNR(\tpii)$.
\end{compactenum}
\end{thm}

To put the above theorem in context, note  that evaluating $\enemySNR(\tp)$ is non-trivial due to multiple levels of dynamics and dependencies:
$\tp$ determines our 
probe signal; then the adversary observes our probe signal in noise, computes the posterior (using the Bayesian  filter (\ref{eq:belief}) with transition matrix $\tp$) and then responds with an action $\laction_k$ generated according to (A2) which we observe in noise.
Despite this complexity, the theorem says that  we can impose a partial order on the probe transition matrices $\tp$ resulting in a corresponding
ordering of the SNR of our observation of the adversary's action (which is a surrogate for the covariance of our estimate of $\oprob$).

\subsubsection{Examples}
We start with   simple 2-state examples    of (A1) and copositive dominance $\tpi \lR \tpii$ for illustrative purposes; see \cite{Kri16} for several high dimensional examples:
    \begin{equation*}
      \begin{split}
 (i):  \tpi &= \begin{bmatrix}  0.8 & 0.2 \\ 0.7 & 0.3
    \end{bmatrix}, \quad \tpii = \begin{bmatrix} 0.2 & 0.8 \\ 0.1 & 0.9 
    \end{bmatrix} \\
 (ii):   \tpi &= \begin{bmatrix} 1 & 0 \\ 1 - p & p 
    \end{bmatrix}, \quad
    \tpii = \begin{bmatrix} 1 & 0 \\ 1 - q & q 
    \end{bmatrix}, \quad q > p
  \end{split}
\end{equation*}
For examples (i) and (ii) above,  (A1) and copositive dominance holds; so if (A2) also holds then Theorem~\ref{thm:snr} holds.

We now illustrate Theorem~\ref{thm:snr} with   higher dimensional examples.

{\em Example 1. Optimizing Probe Signal with Stopping Probability  Constraints}.
Let state 1 denote the stopping (absorbing state). The transition matrix of our probe signal is of the form
$$\tp_{i1}(\param) = 1 - (\statedim-1) \param_i, \;\tp_{ij}(\param) = \param_i, \;  \tp_{11}(\param) = 1.$$
where $\param_i \in [0,1/(\statedim-1)]$, $i=2,\ldots,\statedim$ so that $\tp(\param)$ is a valid stochastic matrix. Then (A1) holds.
Also $\tp(\param) \lR \tp(\bar{\param})$ for $\param < \bar{\param}$ elementwise.

If state 1 is viewed as sending no probe signal, then the probe signal stops
(to maintain a low probability of intercept constraint) at a random time whose 
expected value 
  is determined by $\param_i$. The larger $\param_i$ is chosen, the higher the SNR of the adversary.

{\em Example 2. Monotone Likelihood Ratio (MLR)  ordered probability vectors}. Suppose we have $L$ probability vectors $v_i$ each of dimension $\statedim$ that are MLR ordered (see Defn.\ \ref{def:mlr}), i.e.  $v_1\lr v_2 \cdots \lr v_L$.  We need to choose $\statedim$ vectors out of these $L$ vectors to construct the transition  matrix of our probe signal. Then the optimal probability vectors  to choose for the transition matrix are $\tpii  = \begin{bmatrix}  v_{L-\statedim+1},\cdots,v_{L-1},v_L \end{bmatrix}^\p$. It can be shown that (A1) holds and $\tpii$ copositively dominates any other  transition matrix $\tpi$ constructed from  vectors $v_1,\ldots, v_{L- \statedim}$. Thus Theorem
\ref{thm:snr} implies that $\tpii$ yields the probe signal with the optimal SNR.

In summary, this subsection presented a partially ordering for the transition  matrices of probe signals which resulted in a corresponding ordering
of SNRs (surrogate for the covariance of our estimate of the adversary's sensor gain).

\subsection{Stochastic Gradient Algorithm for Optimizing Probe}
The previous subsection presented an analytical framework for characterizing the probe signals. Here
 we present a numerical stochastic gradient algorithm to  solve the following stochastic optimization problem:
Given our state sequence $\state_{0:\finaltime}$ and observations of the adversary's actions $\action_{1:\finaltime}$, find  the  transition matrix $\tpopt$ of our probe signal that minimizes the variance  of our estimate of the adversary's likelihood:
\beq \tpopt = \argmin_{\tp} \objective(\tp) \ole \E_{\state_{0:\finaltime},\action_{1:\finaltime}}\{ \var(\hat{\oprob}) \}
\label{eq:covlikelihood} \eeq
Here $\hat{\oprob}$ is our MLE of $\oprob$ given $\action_{1:\finaltime}$ computed using the methods in Sec.\ref{sec:mle},
and $\var(\hat{\oprob})$ denotes the empirical variance.

Note that  (\ref{eq:covlikelihood}) is a constrained optimization problem
since $\tp$ needs to be a   valid stochastic matrix. We use  the following trick to re-parametize the problem as an unconstrained
optimization. We parametrize
$\tp$  using spherical coordinates: Let $\param_{ij} \in \reals$ for
$i\in \states$, $j=\{1,\ldots,\statedim-1\}$ and define the parametrized transition  matrix  $\tp(\param)$ as
\begin{equation}
 \begin{split} \tp_{i1}(\param) = \cos^2\param_{i1}, \quad
   \tp_{i\statedim} (\param)= \sin^2 \param_{i,\statedim-1}
   \prod_{l=1}^{\statedim-2} \sin^2 \param_{il} \\
   \tp_{ij}(\param) = \cos^2\param_{ij} \, \prod_{l=1}^{j-1} \sin^2 \param_{il} , \;j = 2,\ldots,\statedim-1 
 \end{split} \label{eq:spherical}
\end{equation}
Note that even though $\param$ is unconstrained,  is straightforwardly  verified that $\tp(\param)$ is a valid stochastic  matrix.\footnote{This approach is based on elementary differential geometry and is equivalent to constraining the derivative to a manifold that ensures the stochastic gradient update yields a valid transition matrix \cite{HM12}. For $\statedim=2$, the parametrization yields $\begin{bmatrix} \cos^2 \param_{11} & \sin^2 \param_{11} \\  \cos^2\param_{21}  & \sin^2 \param_{21} 
  \end{bmatrix}$ which is always a valid transition matrix. Another possibility is to use logistic functions.}

The stochastic optimization (\ref{eq:covlikelihood}) can be solved using a    stochastic gradient algorithm. The key issue is to estimate
the  gradient of the covariance of the MLE $\oprob^*$ wrt $\param$. 
Algorithm \ref{alg:spsa} 
uses one possible stochastic gradient algorithm, namely,  the Simultaneous Perturbation Stochastic Approximation (SPSA) algorithm
\cite{Spa03}
to
generate a sequence of estimates $\param^\iter$,
 $\iterI=1,2,\ldots,$ that
 converges to a local stationary point of the objective $\objective(\tp)$ defined in
 (\ref{eq:covlikelihood}).

\begin{algorithm}[h]
\caption{SPSA  Algorithm for estimating optimal probe transition matrix $\tp^*$ parametrized by $\param$ in (\ref{eq:spherical})} \label{alg1}

Step 1: Choose initial parameter     $\param^{(0)}$.
\\
Step 2: For iterations $\iterI=0,1,2,\ldots$ 
\begin{itemize}
  \item Simulate system (\ref{eq:model}) and generate action sequence $\action_{1:\finaltime}^\iter = (\action_1,\ldots,\action_\finaltime)^\iter$.
\item Evaluate  sample variance $\objective^\iter $ of MLE of $\oprob$  using
 $\action_{1:\finaltime}^\iter$.
\item Compute gradient estimate  
  \begin{equation*}
    \begin{split}
  \nablap \objective^\iter &= \frac{
\objective^\iter(\param^\iter+\Delta_\iterI \direction_\iterI) - \objective^\iter(\param^\iter-\Delta_\iterI \direction_\iterI)}{\displaystyle 2
\Delta_\iterI } \direction_\iterI, \\
\direction_\iterI(i) &= \begin{cases}
-1 & \text{ with probability } 0.5 \\
+1 & \text{with probability } 0.5.\end{cases}.
\end{split}
\end{equation*}
 Here $\Delta_n= \Delta/(n+1)^\gamma$ denotes the gradient step size with $0.5 \leq \gamma \leq 1$ and $\Delta > 0$.

\item Update estimate via 
stochastic gradient  algorithm
\begin{equation} \label{eq:sa}
\begin{split}
\param^\iterplus &= \param^\iter - \epsilon_{n+1} \nablap  {J}^\iter(\param^\iter), \\
\epsilon_n &= \epsilon/(n+1+s)^\zeta, \quad \zeta \in (0.5,1], \;
\epsilon , s > 0. 
\end{split}
\end{equation}
\end{itemize} \label{alg:spsa}
\end{algorithm}

The  SPSA Algorithm~\ref{alg:spsa} \cite{Spa03} 
picks a single
random direction $\direction_n$ along which direction the derivative is
evaluated at each batch $n$.  
To evaluate the gradient estimate $\nablap J^\iter$ in (\ref{eq:sa}), SPSA
requires only 2 batch simulations, i.e., the number of evaluations is
independent of dimension of  parameter  $\param$.
Algorithm \ref{alg1} converges
to a local stationary  point optimum of objective $\objective(\tp)$
(defined in (\ref{eq:covlikelihood})) with probability one. This is a straightforward application of techniques in  \cite{KY03}  (which gives general convergence methods
for Markovian dependencies).

\section{Conclusions}
This paper is an early  step towards understanding
adversarial signal processing problems in counter-autonomous systems.
We presented four main results:
First, we formulated the problem of estimating the  adversary's filtered density (posterior) given noisy measurements (adversary's actions), when the adversary itself runs a Bayesian filter to estimate our state.  Several such inverse filtering algorithms were proposed for various signal models.
Second, we discussed sequential localization as a game between us and the adversary.  The convergence of the inverse filter is slower compared to the adversary's Bayesian filter. An interesting result was the  phase transition in the asymptotic covariance. 
Third, 
we discussed how to compute the maximum likelihood estimate of the adversary's observation matrix using the inverse filter. This allows us to calibrate the accuracy of the adversary's sensors and therefore estimate the adversary's sensor capability. 
Finally, we discussed how to adapt (control)  our state   to minimize the covariance of the estimate of the adversary's observation likelihood. ``Our'' state can be viewed as a probe signal which causes the adversary to act; so choosing the optimal state sequence is an input design problem.

Estimating the posterior distribution and sensor gain of the adversary is a difficult problem. In general identifiability issues can arise for estimating the adversary's sensor gain. Also geometric ergodicity (stability) of the inverse filter, that is, the inverse filter forgets its initial condition geometrically fast, is difficult to establish.
There are several interesting extensions that can be considered in future work.
One area is to study  the convergence properties of the maximum likelihood estimate to the CRB for the inverse filtering case, i.e.,  sample complexity results for the error variance to  converge to the CRB. Another extension is to estimate  the utility function of the adversary based on its response, using for example, revealed preferences from microeconomics. Finally, a careful dynamic game theoretic formulation of counter-autonomous systems is of interest, where the specific forward and backward physical channel responses are considered.

{\bf Acknowledgment}.  Cornell PhD student Kunal Pattanayak did  the simulations for  generating Figure \ref{fig:kf} and Table \ref{tab:cr}. We thank
Dr.\ Erik Blasch for several technical discussions. This work extends our previous papers to a Bayesian framework; we are grateful to Robert Mattila,
 Cristian Rojas and Bo Wahlberg at KTH for preliminary discussions.

\appendix

              \section*{Proof of Theorem \ref{thm:gamecov}}
              
Our update dynamics via the  Kalman filter are:
\begin{compactitem}
\item At odd time instants, using effective observation  $\eobs_k$ (from the adversary)
our  state estimate and covariance evolves as
  \begin{align}
 \hstate_k &= (1 - \enemykg_k) \hstate_{k-1} + \enemykg_k \eobs_k
             \nonumber \\
 \kalmancov_k^{-1}&= \kalmancov_{k-1}^{-1} + \enemyonoisecov_{k}^{-1}\nonumber \\
       \text{ where } &  \enemykg_k = \frac{\kalmancov_{k-1}}
                        {\kalmancov_{k-1} +    \enemyonoisecov_k},
\quad   
                        \kg_k = (1 + \kalmancov_{k-1}^{-1})^{-1}
                        \label{eq:step1}\\
 \enemyonoisecov_k &=  \frac{(1 - \kg_k)^2}{\kg_k^2} \, \anoisecov+ \frac{\anoisecov}{\kg_k^2} + 1,
                       \nonumber
  \end{align}
 Note:  Even though $\eobs_k$ has $\anoise_k$ and $\anoise_{k-1}$, since the update only happens every two time steps (odd times), they are statistically independent of
  $\anoise_{k+2}$ and $\anoise_{k+1}$ in $\eobs_{k+2}$.

\item At even time instants, using our own observation $\obs_k$,
  \beq
  \begin{split} \hstate_k &= (1 - \kg_k) \hstate_{k-1} + \kg_k \obs_k , \quad \kg_k = \kalmancov_k\\
    \kalmancov_k^{-1}  &= \kalmancov_{k-1}^{-1} + 1
  \end{split} \label{eq:step2}
  \eeq
\end{compactitem}
To analyze the convergence rate, consider  the precision matrix
$\precision_k = \kalmancov_k^{-1}$.
The combined covariance update of (\ref{eq:step1}),  (\ref{eq:step2})  is
\begin{multline}
\precision_{k+2} 
  = \precision_{k}+ 1 \\
  +\frac{1}{1 + \anoisecov\, (2+\precision_{k})^2 \big[ 1 + (1 - (\precision_k+2)^{-1})^2\big]}\label{eq:combined}
     \end{multline}
Thus asymptotically $\precision_k = O(k)$ and
$\kalmancov_k = O(k^{-1})$. However, the initial behavior of the covariance exhibits slower decrease due to the last term in (\ref{eq:combined}).
It is straightforward (but tedious)  to show that for large $k$, $\kalmancov_k = 2/k $.

\section*{Proof of Theorem  \ref{thm:snr}}
Let $\belieftpi_k$ and $\belieftpii_k$,  denote the posteriors computed using Bayesian filter (\ref{eq:belief}) with
transition matrices $\tpi$ and $\tpii$ respectively, initialized  with common prior $\belief_0$.
Under (A1),  it follows from \cite[Theorem 10.6.1]{Kri16} that copositive dominance $\tpi \lR \tpii$ implies  the following sample path likelihood ratio dominance:
$\belieftpi_k \lr \belieftpii_k$ for all time $k$ where $\lr$ denotes monotone likelihood  ratio dominance.  Since likelihood ratio dominance implies first order stochastic dominance \cite{MS02},   it follows that under (A2) that $\levels^\p \belieftpi_k \leq
\levels^\p \belieftpii_k$.  From  (A2) $\fun(\cdot)$ is non-negative and increasing, and so
$\laction^{\tpi}_k \leq \laction_k^{\tpii}$ implying $\E\{(\laction^{\tpi}_k)^2\} \leq
\E\{(\laction^{\tpii}_k)^2\}$.
              
\bibliographystyle{IEEEtran}

\bibliography{C:/Users/vikramk/styles/bib/vkm}

\end{document}